\documentclass[final]{amsart}
\usepackage{microtype}
\usepackage[foot]{amsaddr}

\usepackage{amsmath}
\usepackage{amsfonts}
\usepackage{amssymb}
\usepackage{amsthm}
\usepackage{graphicx}
\usepackage{color}
\definecolor{jkl}{rgb}{0, 0, 0.4}
\usepackage{complexity}
\usepackage{xcolor}
\definecolor{dblue}{rgb}{0,0,0.3}
\usepackage{hyperref}
\usepackage{enumitem}
\usepackage{bm}
\usepackage{mathtools}
\usepackage{leftidx}
\usepackage{isomath}
\usepackage{cite}

\usepackage{algorithm}
\usepackage{algpseudocode}

\newcommand{\remove}[1]{}

\newtheorem{theorem}{Theorem}[section]
\newtheorem{definition}{Definition}
\newtheorem{proposition}{Proposition}[section]
\newtheorem{lemma}{Lemma}
\newtheorem{claim}{Claim}[theorem]
\newtheorem{corollary}{Corollary}[section]
\newtheorem{example}{Example}
\newtheorem{remark}{Remark}

\newcommand{\n}{\{1,\ldots,n\}}

\newcommand{\vbl}{\textrm{vbl}}
\newcommand{\pr}{{\rm pr}}
\newcommand{\rMod}{\textrm{Mod}}
\newcommand{\Cl}{\mathcal{C}}
\newcommand{\scc}{\textrm{scc}}

\newcommand{\lpd}{\textrm{lpd}}
\newcommand{\lpic}{\textrm{lpic}}
\newcommand{\maj}{{\rm maj}}
\newcommand{\boplus}{\bar{\oplus}}

\title[Alg. Efficient Syntactic Char. of  Possibility Domains]{Algorithmically Efficient  Syntactic Characterization of  Possibility Domains}

\author[J. D\'{i}az]{Josep D\'{i}az$^1$}
\thanks{The first two authors' research was partially supported by TIN2017-86727-C2-1-R, GRAMM}
\email{diaz@cs.upc.edu}

\author[L. Kirousis]{Lefteris Kirousis$^{1,2}$}
\thanks{The research of the second author was carried  out while   visiting the Computer Science Department of the Universitat Polit\`{e}cnica de Catalunya}

\author[S. Kokonezi]{Sofia Kokonezi$^2$}

\author[J. Livieratos]{John Livieratos$^2$}
\address{\hspace{-0.45cm}$^1$Computer Science Department, Universitat Polit\`{e}cnica de Catalunya, Barcelona\\ $^2$Department of Mathematics, National and Kapodistrian University of Athens}
\email{\{lkirousis, jlivier89, skoko\}@math.uoa.gr}

\begin{document}

\begin{abstract} In the field of Judgment Aggrgation, a \emph{domain}, that is a subset of a Cartesian power of $\{0,1\}$, is considered to reflect abstract rationality restrictions on vectors of two-valued judgments on a number of issues. We are interested in the ways we can aggregate the positions of a set of individuals, whose positions over each issue form vectors of the domain, by means of unanimous (idempotent) functions, whose output is again an element of the domain. Such functions are called \emph{non-dictatorial}, when their output is not simply the positions of a single individual. Here, we consider domains admitting various kinds of non-dictatorial aggregators, which reflect various properties of majority aggregation: (locally) non-dictatorial, generalized dictatorships, anonymous, monotone, StrongDem and systematic. We show that interesting and, in some sense, democratic voting schemes are always provided by domains that can be described by propositional formulas of specific syntactic types we define. Furthermore, we show that we can efficiently recognize such formulas and that, given a domain, we can both efficiently check if it is described by such a formula and, in case it is, construct it. Our results fall in the realm of classical results concerning the syntactic characterization of domains with specific closure properties, like domains closed under logical AND which are the models of Horn formulas. The techniques we use to obtain our results draw from judgment aggregation as well as propositional logic and universal algebra. \end{abstract}

\maketitle

\section{Introduction}\label{sec:intro} We call {\em domain} any arbitrary subset of a Cartesian power  $\{0,1\}^n$ ($n \geq 1$) when we think of it as the set of  yes/no ballots, or accept/reject judgment vectors on $n$ issues that are  ``rational'', in the sense manifested by being a member of the subset. A domain $D$ has a non-dictatorial aggregator if for some $k \geq 1$ there is a unanimous (idempotent) function $F:D^k \rightarrow D$ that is not a projection function. Such domains are called \emph{possibility domains}.
 The theory of judgment aggregation  was put in this abstract framework   by Wilson \cite{wilson1975theory}, and then elaborated  by several others (see e.g. the work by Dietrich \cite{dietrich2007generalised} and  Dokow and Holzman \cite{dokow2010aggregation, dokow2009aggregation}).  It can be trivially shown that non-dictatorial aggregators always exist unless we demand that $F$ is defined on an issue by  issue fashion (see next section for formal definitions). Such aggregators are called Independent of Irrelevant Alternatives (IIA). In this work aggregators are assumed to be IIA.  

It is a well known  fact from elementary Propositional Logic  that for every subset $D$ of $\{0,1\}^n$, $n \geq 1$,  i.e.  for  every domain, there is a Boolean formula in Conjunctive Normal Form (CNF) whose set of satisfying truth assignments, or  models, denoted by $\textrm{Mod}(\phi)$, is equal to $D$ (see e.g. Enderton \cite[Theorem 15B]{enderton2001mathematical}). Zanuttini and H\'{e}brard  \cite{zanuttini2002unified} give  an algorithm that finds such a  formula and runs in polynomial-time with respect to  the size of the representation of $D$ as input. Following Grandi and Endriss \cite{grandi2013lifting}, we call such a $\phi$ an {\em integrity constraint}  and think of it as expressing the ``rationality'' of $D$ (the term comes from databases, see e.g. \cite{elmasri2016fundamentals}).  

We prove that a domain is a possibility domain,  if and only if it admits  an  integrity constraint  of a certain syntactic form to be precisely defined, which we call a {\em possibility} integrity constraint.   Very roughly, possibility integrity constraints  are formulas that belong to one of three types, the first two of which correspond to ``easy'' cases of possibility domains: (i) formulas whose variables can be partitioned into two non-empty subsets so that no clause contains variables from both sets that we call \emph{separable} and (ii) formulas whose clauses are  exclusive OR's of their literals (\emph{affine} formulas). The most interesting third type is comprised of formulas such that if we change the logical sign of some of their variables, we get formulas that have a Horn part and whose remaining clauses contain only negative occurrences of the variables in the Horn part. We call such formulas  {\em renamable partially Horn}, whereas we call {\em partially Horn}\footnote{A weaker notion of Horn formulas has appeared before in the work of Yamasaki and Doshita \cite{yamasaki1983satisfiabilty}; however  our notion   is incomparable with theirs, in the sense that the class of partially Horn formulas in neither a subset nor a superset (nor equal) to the class $\mathbb{S}_0$ they define.} the  formulas that belong to the third type without having to rename any variables. Furthermore, we show that the unified framework of Zanuttini and H\'ebrard \cite{zanuttini2002unified} for producing formulas of a specific type that describe a given domain, and which entails the notion of prime formulas (i.e. formulas that we cannot further simplify its clauses; see Definition \ref{def:prime}) works also in the case of possibility integrity constraints. Actually, in addition to the syntactical characterization of possibility domains, we give two algorithms: the first  on input  a formula decides whether it is a possibility integrity constraint in time linear in the length of the formula (notice that the definition of possibility integrity constraint entails searching over all subsets of variables of the formula);  the  second  on input  a domain $D$ halts in time polynomial in the size of $D$ and either decides that $D$ is not a possibility domain or otherwise returns a possibility integrity constraint that describes $D$.  It should be noted    that the satisfiability problem remains NP-complete even when restricted to formulas that are partially Horn. However in Computational Social Choice,  domains are considered to be non-empty (see paragraph preceding Example \ref{ex:pH}).

We then consider \emph{local possibility domains}, that is, domains admitting IIA aggregators whose components are all different than any projection function. Such aggregators are called \emph{locally non-dictatorial} (see \cite{nehring2010abstract}). Local non-dictatorial domains were introduced in \cite{kirousis2017aggregation} as \emph{uniform possibility domains} (the definition entails also non-Boolean domains). We show that local possibility domains are described by formulas we call \emph{local possibility integrity constraints} and again, we provide a linear algorithm that checks if a formula is a local possibility integrity constraint and a polynomial algorithm that checks if a domain is a local possibility one and, in case it is, constructs a local possibility integrity constraint that describes it. 

There are various notions of non-dictatorial aggregation, apart from the above, that have been introduced in the field of Aggregation Theory. First, we consider domains that admit aggregators which are not \emph{generalized dictatorships}. A $k$-ary aggregator is a generalized dictatorship that, on input any $k$ vectors from a domain $D$, always returns one of those vectors as its output. These aggregators are a natural generalization of the notion of dictatorial aggregators, in the sense that they select a possibly different ``dictator'' for each set of $k$ feasible voting patters, instead of a single global one. They where introduced by Cariani et al. \cite{cariani2008decision} as \emph{rolling dictatorships}, under the stronger requirement that the above property holds for any $k$ vectors of $\{0,1\}^n$. In that framework, Grandi and Endriss \cite{grandi2013lifting} showed that generalized dictatorships are exactly those functions that are aggregators for every Boolean domain. In this work, we show that domains admitting aggregators which are not generalized dictatorships are exactly the possibility domains (apart from some trivial cases), and are thus described by possibility integrity constraints.

Then, we consider \emph{anonymous} aggregators, which are aggregators that are not affected by permutations of their input and \emph{monotone} aggregators, which are aggregators that do not change their output if a voter changes his choice in order to agree with it. Both of these types of aggregators have been extensively studied in the bibliography (see e.g. \cite{dokow2009aggregation,dokow2010aggregation,nehring2010abstract,grandi2011binary,grandi2013lifting,list2012theory,kirousis2017aggregation}), as they have properties that are considered important, if not necessary, for democratic voting schemes. Here, we show that domains admitting anonymous aggregators are described by local possbibility integrity constraints, while domains admitting non-dictatorial monotone aggregators by separable or renamable partially Horn formulas.    

We also consider another kind of non-dictatorial aggregator that shares an important property of majority voting. \emph{StrongDem} aggregators are $k$-ary aggregators that, on every issue, we can fix the votes of any $k-1$ voters in such a way that the $k$-th voter cannot change the outcome of the aggregation procedure. These aggregators were introduced by Szegedy and Xu \cite{szegedy2015impossibility}. Here, we show that domains admitting StrongDem aggregators are described by a subclass of local possibility integrity constraints.

Finally, we consider aggregators satisfying \emph{systematicity} (see List \cite{list2012theory}). Aggregators are called systematic when they aggregate every issue with a common rule. This property has appeared also as \emph{(issue-)neutrality} in the bibliography (see e.g. Grandi and Endriss \cite{grandi2013lifting} and Nehring and Puppe \cite{nehring2010abstract}). By viewing a domain $D$ as a Boolean relation, systematic aggregators are in fact \emph{polymorphisms} of $D$ (see Section \ref{ssec:syst}). Polymorphisms are a very important and well studied tool of Universal Algebra. Apart from showing, using known results, that domains admitting systematic aggregators are described by specific types of local possibility integrity constraints, we also examine how our previous results concerning the various kinds of non-dictatorial voting schemes are affected by requiring that the aggregators also satisfy systematicity.

It should be mentioned that, to prove our results, we use either implicitly or explicitly what is known as \emph{Post's lattice}. Post \cite{post1941two} completely classified \emph{clones} of Boolean functions, that is sets of Boolean functions containing all the projections, that are closed under \emph{superposition} (see Section \ref{sec:other} for formal definitions). Here, we take advantage of the fact that, when aggregators of a domain are defined in an issue-by-issue fashion, the set of their components on any given issue forms a clone.

As examples of similar classical results in the theory of Boolean relations, we mention that  domains component-wise closed under $\land$ or  $\lor$  have been  identified with the class of domains  that  are  models of Horn or dual-Horn formulas respectively (see Dechter and Pearl  \cite{dechter1992structure}). Also it is known that  a domain  is component-wise closed under the ternary sum \hspace{-1em} $\mod 2$  if and only if it is the  set of   models of a formula that is a conjunction  of subformulas each of which is an  exclusive OR  (the term ``ternary'' refers to the number of bits to be summed). Finally,    a domain  is  closed under the  ternary  majority operator if and only it is the set of models of  a  CNF formula where each clause has at most two literals.  The latter two results are due to Schaefer \cite{schaefer1978complexity}. The  ternary majority operator is the  ternary Boolean function   that returns 1 on input three  bits  if and only if at least two of them are 1. It is also known that the respective formulas for each case can be  found  in polynomial time with respect to the size of $D$ (see Zanuttini and H\'{e}brard \cite{zanuttini2002unified}).

Our results can be interpreted as verifying that various kinds of non-dictatorial voting sche\-mes can always  be generated by  integrity constraints that have a specific, easily recognizable  syntactic form. This can prove valuable for applications in the field of judgment aggregation, where relations are frequently encountered in compact form, as the sets of models of integrity constraints. As examples of such applications, we mention the work of Pigozzi \cite{pigozzi2006belief} in avoiding the \emph{discursive dilemma}, the characterization of \emph{safe agendas} by Grandi and Endriss \cite{grandi2011binary} and that of Endriss and de Haan \cite{endriss2015complexity} concerning the \emph{winner determination problem}. Our proofs draw from results  in judgment aggregation theory as well as from  results   about propositional formulas and logical relations. Specifically, as stepping stones  for our algorithmic syntactic characterization we use  three results. First, a  theorem implicit in Dokow and Holzman \cite{dokow2009aggregation} stating that a domain is a possibility domain if and only if it either admits a binary (of arity 2) non-dictatorial aggregator or it is component-wise closed  under the ternary direct sum. This result was generalized by Kirousis et al. \cite{kirousis2017aggregation} for domains in the non-Boolean framework. Second, a characterization of local possibility domains proven by Kirousis et al. in \cite{kirousis2017aggregation}. Lastly, the ``unified framework for structure identification'' by Zanuttini and  H\'{e}brard \cite{zanuttini2002unified} (see next section for definitions).
\vspace{0.2cm}

\textbf{Relation to the conference version:} A preliminary version of this paper appeared in the Proceedings of the 46th International Colloquium on Automata, Languages and Programming (ICALP 2019) \cite{DBLP:conf/icalp/DiazKKL19}. The  present full version, in addition to detailed proofs and several improvements in the presentation, contains the results about generalized dictatorships, anonymous, monotone and StrongDem aggregators, and the discussion about systematicity, which were not included in the conference version.

\section{Preliminaries}\label{sec:prelim}
We first give the notation and basic definitions from Propositional Logic and judgment aggregation theory that we will use.

Let $V=\{x_1, \ldots, x_n \} $ be a set of Boolean variables. A literal is either a variable $x \in V$ (positive literal)  or a negation $\neg x$ of it (negative literal). 
A clause is a disjunction 
$(l_{i_1} \lor  \cdots \lor l_{i_k})$ 
of literals from different variables. A propositional formula $\phi$  (or just a  ``formula'', without the specification ``propositional'', if clear from the context) in Conjunctive Normal  Form (CNF) is a conjunction of clauses. A formula is called $k$-CNF if every clause of it contains exactly $k$ literals. A (truth) assignment to the variables is an assignment of either 0 or 1 to each of the variables. We denote by $a(x)$ the value of $x$ under the assignment $a$. Truth assignments will be identified with elements of $\{0,1\}^n$, or $n$-sequences of bits. The truth value of a formula for an assignment is computed by the usual rules that apply to logical connectives. The set of satisfying (returning the value 1)  truth assignments, or models, of a formula, is denoted by $\textrm{Mod}(\phi)$. In what follows, we will assume, except if specifically noted, that $n$ denotes the number of variables of a formula $\phi$ and $m$ the number of its clauses.

We say that a variable $x$ appears \emph{positively} (resp. \emph{negatively}) in a clause $C$, if $x$ (resp. $\neg x$) is a literal of $C$. A variable $x\in V$ is positively (resp. negatively) \emph{pure} if it has only positive (resp. negative) appearances in $\phi$.

A Horn clause is a clause with at most one positive literal. A dual Horn is a clause with at most one negative literal. A formula that contains only Horn (dual Horn) clauses is called Horn (dual Horn, respectively). Generalizing the notion of a clause,  we will also call clauses sets of  literals connected with exclusive OR (or direct sum), the logical connective that corresponds to summation in $\{0,1\}  \mod 2$. Formulas obtained by considering a conjunction of  such clauses are called affine. Finally, bijunctive are called the formulas  whose  clauses, in inclusive disjunctive form,  have at most two literals. A domain $D \subseteq \{0,1\}^n $ is called Horn, dual Horn, affine or bijunctive respectively, if there is a Horn, dual Horn, affine or bijunctive formula $\phi$ of $n$ variables such that 
$\textrm{Mod}(\phi)  =D $. In the previous section, we mentioned   efficient solutions to classical syntactic characterization problems for classes of relations with given closure properties on one hand, and formulas of the syntactic forms mentioned above on the other. 

We have presented the above notions and results  without many details, as they are all classical results. For the notions that follow we give more detailed definitions and examples. The first one, as far as we can tell, dates back to 1978 (see Lewis \cite{lewis1978renaming}).

\begin{definition}\label{def:renamable}
A formula $\phi$ whose variables are among the elements of the set $V= \{x_1, \ldots, x_n\}$  is called {\em renamable Horn}, if there is a subset $V_0 \subseteq V$  so that if we replace every appearance  of every negated  literal $l$ from  $V_0$ with the corresponding positive one and vice versa,  $\phi$ is transformed to a Horn formula. 
\end{definition}
The process of replacing the literals of some variables with their logical opposite  ones, is called a \emph{renaming} of the variables of $\phi$. It is straightforward to see that any dual-Horn formula is renamable Horn (just rename all its variables).
\begin{example} 
 Consider the formulas $\phi_1=(x_1\lor x_2\lor \neg x_3)\land (\neg x_1 \lor x_3 \lor x_4) \land (\neg x_2 \lor x_3 \lor \neg x_5)$ and $\phi_2=(\neg x_1\lor x_2\lor x_3\lor x_4)\land (x_1\lor \neg x_2\lor \neg x_3)\land (x_4\lor x_5)$, defined over $V=\{x_1,x_2,x_3,x_4,x_5\}$.

The formula $\phi_1$ is renamable Horn. To see this, let $V_0=\{x_1,x_2,x_3,x_4\}$. By renaming these variables, we get the Horn formula $\phi_1^*=(\neg x_1\lor \neg x_2\lor x_3)\land (x_1 \lor \neg x_3 \lor \neg x_4) \land (x_2 \lor \neg x_3 \lor \neg x_5)$. On the other hand, it is easy to check that $\phi_2$ cannot be transformed into a Horn formula for any subset of $V$, since for the first clause to become Horn, at least two variables from $\{x_2,x_3,x_4\}$ have to be renamed, making the second clause not Horn. \hfill$\diamond$
	\end{example}

	It turns out that whether a formula is renamable Horn can be checked in linear time. There are several algorithms that do that in the literature, with the one of del Val \cite{del20002} being a relatively recent such example. The original non-linear one was given by Lewis \cite{lewis1978renaming}.
	
	We now proceed with introducing several syntactic types of formulas:
	\begin{definition}
	A 	formula is called {\em separable} if its variables  can be partitioned into two non-empty disjoint subsets so that no clause of it contains literals from both subsets.
	\end{definition}
\begin{example} 
 The formula $\phi_3=(\neg x_1\lor x_2\lor x_3)\land (x_1\lor \neg x_2\lor \neg x_3)\land (x_4\lor x_5)$  is separable. Indeed, for the partition $V_1=\{x_1,x_2,x_3\}$, $V_2=\{x_4,x_5\}$ of $V$, we have that no clause of $\phi_3$ contains variables from both subsets of the partition. On the other hand, there is no such partition of $V$ for neither $\phi_1$ nor $\phi_2$ of the previous example. \hfill$\diamond$	\end{example}

	The fact that separable formulas can be recognized in linear time is relatively straightforward (see Proposition  \ref{prop:linsep}  in Section \ref{sec:ident}).
	
	We now  introduce the following notions: 

	\begin{definition}\label{def:ph}
		A formula $\phi$ is called {\em partially Horn} if there is a nonempty subset $V_0\subseteq V$ such that (i) the clauses containing only variables from $V_0$ are Horn and (ii) the variables of $V_0$ appear only negatively (if at all) in a clause containing also variables not in $V_0$.
	\end{definition}
If a formula $\phi$ is partially Horn, then any non-empty subset $V_0\subseteq V$ that satisfies the requirements of Definition \ref{def:ph} will be called an \emph{admissible set of variables}. Also the  Horn clauses that contain variables only from $V_0$ will be called  \emph{admissible clauses} (the set of admissible clauses might be empty). A Horn clause with a variable in $V\setminus V_0$ will be called \emph{inadmissible} (the reason for the possible existence of such clauses will be made clear in the following example).

Notice that a Horn formula is, trivially, partially Horn too, as is a formula that contains at least one negative pure literal. It immediately follows that the satisfiability problem remains NP-complete even when restricted to partially Horn formulas (just add a dummy negative pure literal). However, in Computational Social Choice, domains are considered to be non-empty as a non-degeneracy condition. Actually, it is usually assumed that the projection of a domain to any one of the $n$
issues is the set $\{0,1\}$. 

\begin{example}\label{ex:pH}
We first examine the formulas of the previous examples. $\phi_1$ is partially Horn, since it contains the negative pure literal $\neg x_5$. The Horn formula $\phi_1^*$ is also trivially partially Horn. On the other hand, $\phi_2$ and $\phi_3$ are not, since for every possible $V_0\subseteq \{x_1,x_2,x_3,x_4,x_5\}$, we either get non-Horn clauses containing variables only from $V_0$, or variables of $V_0$ that appear positively in inadmissible clauses.

The formula $\phi_4=(x_1\lor \neg x_2)\wedge (\neg x_1 \vee x_2) \land (\neg x_2\lor \neg x_3)\land (\neg x_1\lor x_3\lor x_4)$ is partially Horn. Its first three clauses are Horn, though the third has to be put in every inadmissible set, since $x_3$ appears positively in the fourth clause which is not Horn. The first two clauses though constitute an admissible set of Horn clauses. Finally, $\phi_5=(x_1\lor \neg x_2)\land (x_2\lor \neg x_3)\land (\neg x_1\lor x_3\lor x_4)$ is not partially Horn. Indeed, since all its variables appear positively in some clause, we need at least one clause to be admissible. The first two clauses of $\phi_5$ are Horn, but we will show that they  both have to be included in an inadmissible set.  Indeed, the second has to belong to every  inadmissible  set since $x_3$ appears positively in the third, not Horn, clause. Furthermore, $x_2$ appears positively in the second clause, which we just showed to belong to every inadmissible set. Thus, the first clause also has to be included in every inadmissible set, and therefore $\phi_5$ is not partially Horn. \hfill$\diamond$\end{example}

Accordingly to the case of renamable Horn formulas, we define:
	\begin{definition}\label{def:rph}
	A formula is called {\em renamable partially Horn }
	if some of its variables can be renamed (in the sense of Definition \ref{def:renamable}) so that it becomes partially Horn.
		\end{definition}
Observe that any Horn, renamable horn or partially Horn formula is trivially renamable partially Horn. Also, a formula with at least one pure positive literal is renamable partially Horn, since by renaming the corresponding variable, we get a formula with a pure negative literal. 		
\begin{example}\label{ex:rph} 
 All formulas of the previous examples are renamable partially Horn: $\phi_1^*$, $\phi_1$  and $\phi_4$ correspond to the trivial cases we discussed above, whereas $\phi_2$, $\phi_3$ and $\phi_5$ all contain the pure positive literal $x_4$. 

Lastly, we examine two more formulas: $\phi_6=(\neg x_1\lor x_2\lor x_3 \vee x_4)\land (x_1\lor \neg x_2\lor \neg x_3)\wedge(\neg x_4 \vee x_5)$ is easily not partially Horn, but by renaming $x_4$ and $x_5$, we obtain the partially Horn formula $\phi_6^*=(\neg x_1\lor x_2\lor x_3 \vee \neg x_4)\land (x_1\lor \neg x_2\lor \neg x_3)\wedge (x_4\vee\neg x_5)$, where $V_0=\{x_4,x_5\}$ is the set of admissible variables. One the other hand, the formula $\phi_7=(\neg x_1\lor x_2\lor x_3)\land (x_1\lor \neg x_2\lor \neg x_3)$ is not renamable partially Horn. Indeed, whichever variables we rename, we end up with one Horn and one non-Horn clause, with at least one variable of the Horn clause appearing positively in the  non-Horn clause. \hfill$\diamond$
	\end{example}
	We prove, by Theorem \ref{thm:linearrph} in  Section \ref{sec:ident}  that   checking whether a formula is renamable partially Horn  can be done   in linear  time in the length of the formula.
\begin{remark}\label{rem:ren}
	Let $\phi$ be a renamable partially Horn formula, and let $\phi^*$ be a partially Horn formula obtained by renaming some of the variables of $\phi$, with $V_0$ being the admissible set of variables. Let also $\mathcal{C}_0$ be an admissible set of Horn clauses in $\phi^*$. We can assume  that only variables of $V_0$ have been renamed, since the other variables are not involved in the definition of being partially Horn. Also, we can assume that  a Horn clause of $\phi^*$ whose variables  appear only in clauses  in $\mathcal{C}_0$ belongs to $\mathcal{C}_0$. Indeed, if not,  we can add it to 	$\mathcal{C}_0$.\hfill$\diamond$
			\end{remark}
	\begin{definition}\label{def:posintcon}
	A formula is called a {\em possibility integrity constraint} if it is either separable, or renamable partially Horn or affine.	
	\end{definition}
	From the above and the fact that checking whether a formula is affine is easy we get Theorem \ref{thm:pic} in Section \ref{sec:ident}, which states that checking whether a formula  is  a possibility integrity constraint can be done in  polynomial time in the size of the formula. 
	
Now, let $V,V'$ be two disjoint sets of variables. By further generalizing the notion of a clause of a CNF formula, we say that a $(V,V')$-\emph{generalized clause} is a clause of the form:
$$(l_1\vee\cdots\vee l_s\vee(l_{s+1}\oplus\cdots\oplus l_t)),$$ where the literal $l_j$ corresponds to variable $v_j$, $j=1,\ldots,t$, $s<t$, $v_1,\ldots,v_s\in V$ and $v_{s+1},\ldots,v_t\in V'$. Such a clause is falsified by exactly those assignments that falsify every literal $l_i$, $i=1,\ldots,s$ and satisfy an even number of literals $l_j$, $j=s+1,\ldots,t$. An affine clause is trivially a $(V,V')$-generalized clause, where all its literals correspond to variables from $V'$. Consider now the following syntactic type of formulas. 

\begin{definition}\label{def:upic}
A formula $\phi$ is a \emph{local possibility integrity constraint (\lpic)} if there are three \emph{pairwise disjoint} subsets $V_0,V_1,V_2\subseteq V$, with $V_0\cup V_1\cup V_2=V$, where no clause contains variables both from $V_1$ and $V_2$ and such that:\begin{enumerate}
    \item by renaming some variables of $V_0$, we obtain a partially Horn formula $\phi^*$, whose set of admissible variables is $V_0$,
    \item any clause contains at most two variables from $V_1$ and
    \item the clauses containing variables from $V_2$ are $(V_0,V_2)$-generalized clauses.
\end{enumerate}  
\end{definition}
\begin{example}\label{ex:upic}
Easily, every (renamable) Horn, bijunctive or affine formula is an \lpic. On the other hand, consider the following possibility integrity constraint: $$\phi_8=(\neg x_1\vee x_2\vee x_3 \vee x_4)\wedge(\neg x_2\vee \neg x_3 \vee \neg x_4).$$ $\phi_8$ is partially Horn, since it has the pure negative literal $\neg x_1$ and thus a possibility integrity constraint. But, it is not an \lpic, since however we define $V_0$, $V_1$, either there will be a variable of $V_0$ with a positive appearence in a non-admissible clause (even after any possible renaming of the variables of $V_0$) and/or there will be a clause with more than two literals from $V_1$.\hfill$\diamond$
\end{example}
By Definition \ref{def:upic}, an lpic $\phi$ over $V$, where $V_0\neq\emptyset$, is a renamable partially Horn formula. Otherwise, if $V_2\neq\emptyset$, $\phi$ is either separable or affine. Finally, if $V_1=V$, $\phi$ is renamable Horn, since every $2$-SAT formula is. Indeed, let $\alpha$ be an assignment satisfying $\phi$ and rename all the variables $x\in V$ such that $\alpha(x)=1$. Then, every clause of $\phi$ either has a positive literal that is renamed, or a negative one that is not renamed.

To end this preliminary discussion about propositional formulas, we consider \emph{prime} formulas. Given a clause $C$ of a formula $\phi$, we say that a \emph{sub-clause} of $C$ is any non-empty clause created by deleting at least one literal of $C$. In Quine \cite{quine1959cores} and Zanuttini and H\'ebrard \cite{zanuttini2002unified}, we find the following definitions:
\begin{definition}\label{def:prime}
A clause $C$ of a formula $\phi$ is a \emph{prime implicate} of $\phi$ if no sub-clause of $C$ is logically implied by $\phi$. Furthermore, $\phi$ is prime if all its clauses are prime implicates of it. 
\end{definition}
In Section \ref{sec:character}, we use this notion in order to efficiently construct formulas whose sets of models is a (local) possibility domain.

We now come to some notions from Social Choice Theory (for an introduction, see e.g. List~\cite{list2012theory}). 
	In the sequel, we will deal with $k$ sequences of  $n$-bit-vectors, each of which  belongs to a fixed domain $D \subseteq \{0,1\}^n$. It is convenient to present such sequences with an $k\times n$ matrix $x^i_j, i=1, \ldots, k, j =1, \ldots, n$ with bits as entries. The rows of this matrix  are denoted by $x^i, i=1, \ldots, k$ and the columns  by $x_j, j= 1, \ldots, n$.  Each  row represents a row-vector of 0/1 decisions on $n$ issues by one of $k$ individuals.  Each column represents the column-vector of the positions of all $k$ individuals on a particular issue.
	
	In the sequel, we will assume that all domains $D\subseteq\{0,1\}^n$ are \emph{non-degenerate}, i.e. for any $j\in\n$, it holds that $D_j=\{0,1\}$, where $D_j$ denotes the projection of $D$ to the $j$-th coordinate. This is a common assumption in Social Choice Theory, which reflects the idea that voting is nonsensical when there is only one option. Consequently, we will also assume that the formulas we consider have non-degenerate domains too.
	
	A domain $D \subseteq \{0,1\}^n$ is said to have a $k$-ary (of arity $k$) unanimous aggregator  if  there exists  a sequence of $n$ $k$-ary Boolean functions $(f_1,\ldots,f_n)$, $f_j: \{0,1\}^k \rightarrow \{0,1\}, j =1, \ldots, n$ such that 
	\begin{itemize}
	\item all $f_j$ are unanimous, i.e  if $b_1 = \cdots = b_k$ are equal bits,   then  $$f_j(b_1, \ldots, b_k) = b_1 = \cdots =  b_k, \mbox{ and }$$
	\item \label{itm:2} if for a matrix $(x^i_j)_{i,j}$ that represents the opinions of $k$ individuals on $n$ issues we have that the row-vectors $x^i \in D$ for all $i=1, \ldots, k$, then $$(f_1(x_1), \ldots, f_n(x_n)) \in D.$$  
	\end{itemize} 
	
	Notice that in the second  bullet  above,  the $f_j$'s are applied to   column-vectors, which have dimension $k$. The $f_j$'s are called the  {\em components} of the aggregator $(f_1,\ldots,f_n)$.  Intuitively, an aggregator  is a sequence of functions that  when applied   onto some rational opinion vectors  of $k$ individuals on $n$ issues, in a issue-by-issue fashion, they return  a row-vector that is still  rational. From now on, we will refer to  unanimous aggregators, simply as aggregators. We will also sometimes say that $F$ is an aggregator, meaning that $F$ is a sequence of $n$ functions $(f_1,\ldots,f_n)$ as above.
	
	The fact that we defined aggregators as $n$-tuples of functions, means that we require that they satisfy a property called \emph{Independence of Irrelevant Alternatives}, in the sense that the way we aggregate an issue, is independent of the way we aggregate the rest. An aggregator $F=(f_1,\ldots,f_n)$, where $f_1=f_2=\ldots=f_n:=f$ is called \emph{systematic}. Notationally, we write $\bar{f}$ to denote the $n$-tuple $(f,\ldots,f)$, where the number of components of $\bar{f}$ always corresponds to the arity of the given domain.
	
An aggregator $F=(f_1, \ldots, f_n)$ for a domain $D$ is called {\em dictatorial} if there is a $d= 1. \ldots, k$ such that  $f_1 = \cdots = f_n = \pr_d^k$, where  $\pr_d^k: (b_1, \dots, b_k) \mapsto b_d$ is the $k$-ary projection function on the $d$'th coordinate. 	

A $k$-ary aggregator is called a {\em projection} aggregator if each of its components is a  projection function  $\pr_d^k$, for some $d=1, \ldots, m$. Notice that it is conceivable to have non-dictatorial aggregators that are projection aggregators.

The only unary (of arity $1$) unanimous function is the \emph{identity} function ${\rm id}:\{0,1\}\mapsto\{0,1\}$, where ${\rm id}(x)=x$, $x\in\{0,1\}$. Thus, there is only one unary aggregator, which is trivially dictatorial, as all of its components equal $\pr_1^1$. Thus, from now on, we will always assume that all the functions we consider have arity at least $2$. A binary (of arity 2) Boolean function $f: \{0,1\}^2 \rightarrow \{0,1\}$ is called {\em symmetric} if for  all pairs of bits $b_1, b_2$, we have that  $f(b_1, b_2) = f(b_2, b_1)$. A binary aggregator is called symmetric if all its components are symmetric. Let us mention here the easily to check fact that the only unanimous binary functions are the $\land$, $\lor$ and the two projection functions $\pr_1^2, \pr_2^2$. Of those four, only the first two are symmetric.
\begin{definition}	
A domain $D$ is called a possibility domain if it has a (unanimous) non-dictatorial aggregator of some arity.
\end{definition}
Notice that the search space for such an aggregator is large, as the arity is not restricted. However, from \cite[Theorem 3.7]{kirousis2017aggregation} (a result that follows from  Dokow and Holzman  \cite{dokow2009aggregation}, but without being explicitly mentioned there),  we can easily get that:
\begin{theorem}[Dokow and Holzman \cite{dokow2009aggregation}]\label{thm:DH}	A domain $D$ is a possibility domain if and only if it admits either: (i) a non-dictatorial binary projection aggregator or (ii) a non-projection binary aggregator (i.e. at least one symmetric component) or (iii) a ternary aggregator all components of which are the binary addition \hspace{-1ex} $\mod 2$.
\end{theorem}

\begin{example}\label{ex:possdom}
Theorem \ref{thm:DH} directly implies that the truth set of any affine formula is a possibility domain. Consider now the formula $\phi_7=(\neg x_1\vee x_2 \vee x_3)\wedge(x_1\vee\neg x_2\vee\neg x_3)$ of Example \ref{ex:rph}. It holds that:$$\rMod(\phi_7)=\{0,1\}^3\setminus\{(1,0,0),(0,1,1)\}.$$ By checking all $4^3$ different triples of binary unanimous operators and since $\rMod(\phi_7)$ is not affine, one can see that $\rMod(\phi_7)$ is an impossibility domain. On the other hand, let $$\phi_9:=(\neg x_1\vee x_2 \vee x_3)\wedge(x_1\vee\neg x_2\vee\neg x_3)\wedge(\neg x_4\vee x_5 \vee x_6)\wedge(x_4\vee\neg x_5\vee\neg x_6).$$ Then, we have that:
$$\rMod(\phi_9)=\rMod(\phi_7)\times\rMod(\phi_7),$$ which is a possibility domain, since every Cartesian product is (see Kirousis et al. \cite[Example $2.1$]{kirousis2017aggregation}). Finally, for: $$\phi_6=(\neg x_1\lor x_2\lor x_3 \vee x_4)\land (x_1\lor \neg x_2\lor \neg x_3)\wedge(\neg x_4 \vee x_5),$$ of Example \ref{ex:rph} we have that:$$\rMod(\phi_6)=(\rMod(\phi_7)\times\{(0,0),(0,1)\})\cup\Big((\{0,1\}^3\setminus\{(1,0,0)\})\times\{(1,1)\}\Big)$$ is a possibility domain, as it admits the binary aggregator $(\pr_1^2,\pr_1^2,\pr_1^2,\vee,\vee)$.\hfill$\diamond$
\end{example}

Nehring and Puppe \cite{nehring2010abstract} defined a type of non-dictatorial aggregators they called \emph{locally non-dictatorial}. A $k$-ary aggregator $(f_1,\ldots,f_n)$ is locally non-dictatorial if $f_j\neq \pr_d^k$, for all $d\in\{1,\ldots,k\}$ and $j=1,\ldots,n$.
\begin{definition}\label{def:upd}
$D$ is a \emph{local possibility domain (\lpd)} if it admits a locally non-dictatorial aggregator. 
\end{definition}
Kirousis et al. \cite{kirousis2017aggregation} introduced these domains as \emph{uniform} non-dictatorial  domains, both in the Boolean and non-Boolean framework and provided a characterization for them. Consider the following ternary operators on $\{0,1\}$: (i) $\wedge^{(3)}(x,y,z):=\wedge(\wedge(x,y),z))$ (resp. for $\vee^{(3)}$), (ii) $\maj$, where $\maj(x,y,z)=1$ if and only if \emph{at least} two elements of its input are $1$ and (iii) $\oplus$, where $\oplus(x,y,z)=1$ if an only if \emph{exactly one or all} of the elements of its input are equal to $1$.
\begin{theorem}[Kirousis et al. \cite{kirousis2017aggregation}, Theorem $5.5$]\label{thm:upd-char}
$D\subseteq\{0,1\}^n$ is a local possibility domain if and only if it admits a ternary aggregator $(f_1,\ldots,f_n)$ such that $f_j\in\{\wedge^{(3)},\vee^{(3)},\maj,\oplus\}$, for $j=1,\ldots,n$.
\end{theorem}
\begin{example}\label{ex:lpd}
Neither $\rMod(\phi_6)$ nor $\rMod(\phi_7)$ of Example \ref{ex:rph}, nor $\rMod(\phi_9)$ of Example \ref{ex:possdom} are local possibility domains, since every aggregator they admit has components that are projection functions. On the other hand, for: $$\phi_{10}=(\neg x_1\vee x_2\vee x_3)\wedge(x_1\vee x_2\vee\neg x_3),$$ we have that:$$\rMod(\phi_{10})=\{0,1\}^3\setminus\{(0,0,1),(1,0,0)\},$$ that is a possibility domain, since it admits $(\wedge,\vee,\wedge)$.\hfill$\diamond$
\end{example}

\section{Identifying (local) possibility integrity constraints}\label{sec:ident}
In this section, we show that identifying (local) possibility integrity constraints can be done in time linear in the length of the input formula. By Definitions \ref{def:posintcon} and \ref{def:upic}, it suffices to show that for separable formulas, renamable partially Horn formulas and lpic's, since the corresponding problem for affine formulas is trivial. 

In all that follows, we assume that we have a set of variables $V:=\{x_1,\ldots,x_n\}$ and a formula $\phi$ defined on $V$ that is a conjunction of $m$ clauses $C_1,\ldots,C_m$, where $C_j=(l_{j_1},\ldots,l_{j_{k_j}})$, $j=1,\ldots,n$, and $l_{j_s}$ is a positive or negative literal of $x_{j_s}$, $s=1,\ldots,k_j$. We denote the set of variables corresponding to the literals of a clause $C_j$ by $\vbl(C_j)$.

We begin with the result for separable formulas:
\begin{proposition}\label{prop:linsep}
There is an algorithm that, on input a formula $\phi$, halts in time linear in the length of $\phi$ and either returns that the formula is not separable, or alternatively produces a \emph{partition} of $V$ in two non-empty and disjoint subsets $V_1,V_2\subseteq V$, such that no clause of $\phi$ contains variables from both $V_1$ and $V_2$.
\end{proposition}
\begin{proof} We construct a graph on the variables of $\phi$, where two such vertices are connected if they appear consecutively in a common clause of $\phi$. The result is then obtained by showing that $\phi$ is separable if and only if $G$ is not connected.

Suppose the variables of each clause are ordered by the indices of their corresponding literals in the clause. Thus, we say that $x_{j_s},x_{j_t}$ are consecutive in $C_j$, if $t=s+1$, $s=1,\ldots,k_j-1$.

Given a formula $\phi$, construct an undirected  graph $G=(V,E)$, where :\begin{itemize}
    \item $V$ is the set of  variables of $\phi$, and 
    \item two vertices are connected if they appear consecutively in a common clause of $\phi$. \end{itemize} It is easy to see that each clause $C_j$, where $\vbl(C_j)=\{x_{j_1},\ldots,x_{j_{k_j}}\}$ induces the path $\{x_{j_1},\ldots,x_{j_{k_j}}\}$ in $G$.
    
For the proof of linearity, notice that  the set of edges can be    constructed in linear time with respect to the length of $\phi$, since we simply need to read once each clause of $\phi$ and connect its consecutive vertices. Also, there are standard techniques to check  connectivity in linear time in the number of edges (e.g. by a \emph{depth-first search} algorithm).

The correctness of the algorithm is derived by noticing that two connected vertices of $G$ cannot be separated in $\phi$. Indeed, consider a path $P:=\{x_r,\ldots,x_s\}$ in $G$ (this need not be a path induced by a clause). Then, each couple $x_t,x_{t+1}$ of vertices in $P$ belongs in a common clause of $\phi$, $t=r,\ldots,s-1$. Thus, $\phi$ is separable if and only if $G$ is not connected. \end{proof}

To deal with renamable partially Horn formulas, we will start with Lewis' idea \cite{lewis1978renaming} of creating, for a formula $\phi$, a 2{\sc Sat} formula $\phi'$ whose satisfiability is equivalent to $\phi$ being renamable Horn. However, here we need to (i) look for a renaming that might transform only some clauses into Horn and (ii) deal with inadmissible Horn clauses, since such clauses can cause other Horn clauses to become inadmissible too.

\begin{proposition}\label{prop:renpartHorn}
For every formula $\phi$, there is a formula $\phi'$ such that $\phi$ is renamable partially Horn if and only if $\phi'$ is satisfiable.\end{proposition}

Before delving into the proof, we introduce some notation.  Assume  that after a renaming of some of the  variables in $V$, we get the partially Horn formula $\phi^*$, with $V_0$ being the admissible set of variables. Let $\mathcal{C}_0 $ be an admissible set of clauses for $\phi^*$. We assume below that only a subset $V^*\subseteq V_0$ has been renamed and that all Horn clauses of $\phi^*$ with variables exclusively from $V_0$ belong to $\mathcal{C}_0 $ (see Remark \ref{rem:ren}). Also, let $V_1:=V\setminus V_0$. The clauses of $\phi^*$, which are in a one to one correspondence with those of $\phi$, are denoted by $C_1^*,\ldots,C_m^*$, where $C_j^*$ corresponds to $C_j$, $j=1,\ldots,m$.
 
\begin{proof}
For each variable $x\in V$, we introduce a new variable $x'$. Intuitively, setting $x=1$ means that $x$ is renamed (and therefore $x\in V^*$), whereas setting $x'=1$ means that $x$ is in $V_0$, but is not renamed. Finally we set both $x$ and $x'$ equal to $0$ in case $x$ is not in $V_0$. Obviously, we should not not allow the assignment $x=x'=1$ (a variable in $V_0$ cannot be renamed and not renamed). Let $V'=V\cup\{x'\mid x\in V\}$.

Consider the formula $\phi'$ below, with variable set $V'$. For each clause $C$ of $\phi$ and for each $x\in\vbl(C)$: if $x$ appears positively in $C$, introduce the literals $x$ and $\neg x'$ and if it appears negatively, the literals $\neg x$ and $x'$. $\phi'$ is the conjunction of the following clauses: for each clause $C$ of $\phi$ and for each two variables $x,y\in\vbl(C)$, $\phi'$ contains the disjunctions of the positive with the negative literals introduced above. Thus:
\begin{itemize}
    \item[(i)] if $C$ contains the literals $x,y$, then $\phi'$ contains the clauses $(x\vee \neg y')$ and $(\neg x'\vee y)$,
    \item[(ii)] if $C$ contains the literals $x,\neg y$, then $\phi'$ contains the clauses $(x\vee \neg y)$ and $(\neg x'\vee y')$ (accordingly if $C$ contains $\neg x, y$) and
    \item[(iii)] if $C$ contains the literals $\neg x,\neg y$, then $\phi'$ contains the clauses $(\neg x\vee y')$ and $(x'\vee \neg y)$.
\end{itemize}
Finally, we add the following clauses to $\phi'$: 
\begin{itemize}
    \item[(iv)] $(\neg x_i \vee \neg x'_i)$, $i=1,\ldots,n$ and
    \item[(v)] $\bigvee_{x\in V'} x$. 
\end{itemize}
The clauses of items (i)--(iv) correspond to the intuition we explained in the beginning. For example, consider the case where a clause $C_j$ of $\phi$ has the literals $x,\neg y$. If we add $x$ to $V_0$ without renaming it, we should not rename $y$, since we would have two positive literals in a clause of $\Cl_0$. Also, we should not add the latter to $V_1$, since we would have a variable of $V_0$ appearing positively in a clause containing a variable of $V_1$. Thus, we have that $x'\rightarrow y'$, which is expressed by the equivalent clause $(\neg x'\vee y')$ of item (ii). The clauses of item (iv) exclude the assignment $x=x'=1$ for any $x\in V$. Finally, since we want $V_0$ to be non-empty, we need at least one variable of $V'$ to be set to $1$.

To complete the proof of Proposition \ref{prop:renpartHorn}, we now proceed as follows.

($\Rightarrow)$ First, suppose $\phi$ is renamable partially Horn. Let $V_0$, $V_1$, $V^*$ and $V'$ as above. Suppose also that $V_0\neq\emptyset$.

Set $a=(a_1,\ldots,a_{2n})$ to be the following assignment of values to the variables of $V'$:
\begin{equation*}
    a(x)=\begin{cases}
    1, & \text{ if }x\in V^*,\\
    0, & \text{ else,}
    \end{cases} \ \text{and} \ a(x')=\begin{cases}
    0, & \text{ if }x\in V^*\cup V_1,\\
    1, & \text{ else,}
    \end{cases}
  \end{equation*}
 for all $x\in V$. To obtain a contradiction, suppose $a$ does not satisfy $\phi'$.

Obviously, the clauses of items (iv) and (v) above are satisfied, by the definition of $a$ and the fact that $V_0$ is not empty.

Now, consider the remaining clauses of items (i)--(iii) above and suppose for example that some $(\neg x\vee y')$ is not satisfied. By the definition of $\phi'$, there exists a clause $C$ which, before the renaming takes place, contains the literals $\neg x, \neg y$ (see item (iii)). Since the clause is not satisfied, $a(x)=1$ and $a(y')=0$, which in turn means that $x\in V^*$ and $y\in V^*\cup V_1$. If $y\in V_1$, $C^*$ contains, after the renaming, a variable in $V_1$ and a positive appearance of a variable in $V_0$. If $y\in V^*$, $C^*$ contains two positive literals of variables in $V_0$. Contradiction. The remaining cases can be proven analogously and are left to the reader.

($\Leftarrow$) Suppose now that $a=(a_1,\ldots,a_{2n})$ is an assignment of values to the variables of $V'$ that satisfies $\phi'$. We define the following subsets of $V'$: \begin{itemize}
    \item[-] $V^*=\{x\mid a(x)=1\}$,
    \item[-] $V_0=\{x\mid a(x)=1\text{ or }a(x')=1\}$ and
    \item[-] $V_1=\{x\mid a(x)=a(x')=0\}.$
\end{itemize}
Let $\phi^*$ be the formula obtained by $\phi$, after renaming the variables of $V^*$.

Obviously, $V_0$ is not empty, since $a$ satisfies the clause of item (v).

Suppose that a clause $C^*$, containing only variables from $V_0$, is not Horn. Then, $C^*$ contains two positive literals $x,y$. If $x,y\in V_0\setminus V^*$, then neither variable was renamed and thus $C$ also contains the literals $x,y$. This means that, by item (i) above, $\phi'$ contains the clauses $(x\vee\neg y')$ and $(\neg x'\vee y)$. Now, since $x,y\in V_0\setminus V^*$, it holds that $a(x)=a_(y)=0$ and $a(x')=a(y')=1$. Then, $a$ does not satisfy these two clauses. Contradiction. In the same way, we obtain contradictions in cases that at least one of $x$ and $y$ is in $V^*$.

Finally, suppose that there is a variable $x\in V_0$ that appears positively in a clause $C^*\notin\Cl_0$. Let $y\in V_1$ be a variable in $C^*$ (there is at least one such variable, lest $C^*\in\Cl_0$). Suppose also that $y$ appears positively in $C^*$.

Assume $x\in V^*$. Then, $C$ contains the literals $\neg x, y$. Thus, by item (ii), $\phi'$ contains the clause $(\neg x \vee y)$. Furthermore, since $x\in V^*$, $a(x)=1$ and since $y\in V_1$, $a(y)=0$. Thus the above clause is not satisfied. Contradiction. In the same way, we obtain contradictions in all the remaining cases.
\end{proof}

To compute $\phi'$ from $\phi$, one would need quadratic time in the length of $\phi$. Thus, we introduce the following linear algorithm that decides if a formula $\phi$ is renamable partially Horn, by tying a property of a graph constructed based on $\phi$, with the satisfiability of $\phi'$.

\begin{theorem}\label{thm:linearrph}
There is an algorithm that, on input a formula $\phi$, halts in time linear in the length of $\phi$ and either returns that $\phi$ is not renamable partially Horn or alternatively produces a subset $V^*\subseteq V$ such that the formula  $\phi^*$  obtained from $\phi$ by renaming the literals of variables in $V^*$ is partially Horn.\end{theorem}

To prove Theorem \ref{thm:linearrph}, we define a \emph{directed bipartite} graph $G$, i.e. a directed graph whose set of vertices is partitioned in two sets such that no vertices belonging in the same part are adjacent. Then, by computing its \emph{strongly connected components (\scc)}, i.e. its maximal sets of vertices such that every two of them are connected by a directed path, we show that at least one of them is not \emph{bad} (does not contain a pair of vertices we will specify below) \emph{if and only if} $\phi$ is renamable partially Horn.

For a directed graph $G$, we will denote a directed edge from a vertex $u$ to a vertex $v$ by $(u,v)$. A (directed) path from $u$ to $v$, containing the vertices $u=u_0,\ldots,u_s=v$, will be denoted by $(u,u_1,\ldots,u_{s-1},v)$ and its existence by $u\rightarrow v$. If both $u\rightarrow v$ and $v\rightarrow u$ exist, we will sometimes write $u\leftrightarrow v$.

Recall that given a directed graph $G=(V,E)$, there are known algorithms that can compute the \scc \ of $G$ in time $O(|V|+|E|)$, where $|V|$ denotes the number of vertices of $G$ and $|E|$ that of its edges. By identifying the vertices of each \scc, we obtain a \emph{directed acyclic graph (DAG)}. An ordering $(u_1,\ldots,u_n)$ of the vertices of a graph is called \emph{topological} if there are no edges $(u_i,u_j)$ such that $i\geq j$, for all $i,j\in\n$.   
\begin{proof}
Given $\phi$ defined on $V$, whose set of clauses is $\Cl$ and let again $V'=V\cup\{x'\mid x\in V\}$. We define the graph $G$, with vertex set $V'\cup\Cl$ and edge set $E$ such that, if $C\in\Cl$ and $x\in\vbl(C)$, then:
\begin{itemize}
    \item if $x$ appears \emph{negatively} in $C$, $E$ contains $(x,C)$ and $(C,x')$,
    \item if $x$ appears \emph{positively} in $C$, $E$ contains $(x',C)$ and $(C,x)$ and
    \item $E$ contains no other edges.
\end{itemize}
Intuitively, if $x,y\in V'$, then a path $(x,C,y)$ corresponds to the clause $x\rightarrow y$ which is logically equivalent to $(\neg x\vee y)$. The intuition behind $x$ and $x'$ is exactly the same as in Proposition \ref{prop:renpartHorn}. We will thus show that the bipartite graph $G$ defined above, contains all the necessary information to decide if $\phi'$ is satisfiable, with the difference that $G$ can obviously be constructed in time linear in the length of the input formula.

There is a slight technicality arising here since, by the construction above, $G$ always contains either the path $(x,C,x')$ or $(x',C,x)$, for any clause $C$ and $x\in\vbl(C)$, whereas neither $(\neg x\vee x')$ nor $(x\vee \neg x')$ are ever clauses of $\phi'$. Thus, from now on, we will assume that no path can contain the vertices $x$, $C$ and $x'$ or $x'$, $C$ and $x$ \emph{consecutively}, for any clause $C$ and $x\in\vbl(C)$.

Observe that by construction, (i) $(x,C)$ or $(C,x)$ is an edge of $G$ if and only if $x\in\vbl(C)$, $x\in V'$ and (ii) $(x,C)$ (resp. $(x',C)$) is an edge of $G$ if and only if $(C,x')$ (resp.$(C,x)$) is one too.

 We now prove several claims concerning the structure of $G$. To make notation less cumbersome, assume that for an $x\in V$, $x''=x$. Consider the formula $\phi'$ of Proposition \ref{prop:renpartHorn}.
\begin{claim}\label{app-claim:correctness}
Let $x,y\in V'$. For $z_1,\ldots,z_k\in V'$ and $C_1,\ldots,C_{k+1}\in\Cl$, it holds that $(x,C_1,z_1,C_2,\ldots,z_k,C_{k+1},y)$ is a path of $G$ if and only if $(\neg x \vee z_1)$, $(\neg z_i\vee z_{i+1})$, $i=1,\ldots,k-1$ and $(\neg z_k\vee y)$ are all clauses of $\phi'$. 
\end{claim}
\textit{Proof of Claim.} Can be easily proved inductively to the length of the path, by recalling that a path $(u,C,v)$ corresponds to the clause $(\neg u \vee v)$, for all $u,v\in V'$ and $C\in\Cl$.\hfill$\Box$
\begin{claim}\label{app-claim:reachability}
Let $x,y\in V'$. If $x\rightarrow y$, then $y'\rightarrow x'$. 
\end{claim}
\textit{Proof of Claim.} Since $x\rightarrow y$, there exist $z_1,\ldots,z_k\in V'$ and $C_1,\ldots,C_{k+1}\in\Cl$, such that $(x,C_1,z_1,C_2,\ldots,z_k,C_{k+1},y)$ is a path of $G$. By Claim \ref{app-claim:correctness}, $(\neg x \vee z_1)$, $(\neg z_i\vee z_{i+1})$, $i=1,\ldots,k-1$ and $(\neg z_k\vee y)$ are all clauses of $\phi'$. By Proposition \ref{prop:renpartHorn}, so do $(\neg y' \vee z'_k)$, $(\neg z'_{i+1}\vee z'_i)$, $i=1,\ldots,k-1$ and $(\neg z'_1\vee x')$ and the result is obtained by using Claim \ref{app-claim:correctness} again.\hfill $\Box$
\vspace{0.2cm}

We can obtain the \scc's \ of $G$ using a variation of a \emph{depth-first search (DFS)} algorithm, that, whenever it goes from a vertex $x$ (resp. $x'$) to a vertex $C$, it cannot then go to $x'$ (resp. $x$) at the next step. Since the algorithm runs in time linear in the number of the vertices and the edges of $G$, it is also linear in the length of the input formula $\phi$.

 Let $S$ be a \scc \ of $G$. We say that $S$ is \emph{bad}, if, for some $x\in V$, $S$ contains both $x$ and $x'$. We can decide if each of the \scc's is bad or not again in time linear in the length of the input formula.
\begin{claim}\label{app-claim:badscc}
Let $S$ be a bad \scc \ of $G$ and $y\in V'$ be a vertex of $S$. Then, $y'$ is in $S$.
\end{claim}
\textit{Proof of Claim.} Since $S$ is bad, there exist two vertices $x,x'$ of $V'$ in $S$. If $x=y$ we have nothing to prove, so we assume that $x\neq y$. Then, we have that $y\rightarrow x$, which, by Claim \ref{app-claim:reachability} implies that $x'\rightarrow y'$. Since $x\rightarrow x'$, we get that $y\rightarrow y'$. That $y'\rightarrow y$ can be proven analogously.\hfill$\Box$.

Let the \scc's \ of $G$, in \emph{reverse topological order}, be $S_1,\ldots.S_t$. We describe a process of assigning values to the variables of $V'$:\begin{enumerate}
    \item Set every variable that appears in a bad \scc \ of $G$ to $0$.
    \item For each $j=1,\ldots,t$ assign value $1$ to every variable of $S_j$ that has not already received one (if $S_j$ is bad no such variable exists). If some $x\in V'$ of $S_j$ takes value $1$, then assign value $0$ to $x'$.
    \item Let $a$ be the resulting assignment to the variables of $V'$.
\end{enumerate}
Now, the last claim we prove is the following:
\begin{claim}\label{app-claim:satisf}
There is at least one variable $z\in V'$ that does not appear in a bad \scc \ of $G$ if and only if $\phi'$ is satisfiable.
\end{claim}
\textit{Proof of Claim.} ($\Rightarrow)$ We prove that every clause of type (i)--(v) is satisfied. First, by the construction of $a$, every clause $\neg x_i\vee\neg x'_i$, $i=1,\ldots,n$, of type (iv) is obviously satisfied. Also, since by the hypothesis, $z$ is not in a bad \scc, it holds, by step 2 above, that either $z$ or $z'$ are set to $1$. Thus, the clause $\bigvee_{x\in V'}x$ of type (v) is also satisfied.

Now, suppose some clause $(x \vee \neg y')$ (type (i)) of $\phi'$ is not satisfied. Then $a(x)=0$ and $a(y')=1$. Furthermore, there is a vertex $C$ such that $(y',C)$ and $(C,x)$ are edges of $G$. By the construction of $G$, $(x',C)$ and $(C,y)$ are also edges of $G$.

Since $a(x)=0$, it must hold either that $x$ is in a bad \scc \ of $G$, or that $a(x')=1$. In the former case, we have that $x\rightarrow x'$, which, together with $(y',C,x)$ and $(x',C,y)$ gives us that $y'\rightarrow y$. Contradiction, since then $a(y')$ should be $0$. In the latter case, we have that there are two \scc's $S_p$, $S_r$ of $G$ such that $x\in S_p$, $x'\in S_r$ and $p<r$ in their topological order. But then, there is some $q:p\leq q\leq r$ such that $C$ in $S_q$. Now, if $p=q$, we obtain a contradiction due to the existence of $(x',C)$, else, due to $(C,x)$.   

The proof for the rest of the clauses of types (i)--(iii) are left to the reader.

($\Leftarrow$) First, recall that for two propositional formulas $\phi, \psi$, we say that $\phi$ \emph{logically entails} $\psi$, and write $\phi\models\psi$, if any assignment that satisfies $\phi$, satisfies $\psi$ too.

Now observe that, if $x,y$ are two vertices in $V'$ such that $x\rightarrow y$, then $\phi' \models (\neg x \vee y)$. Indeed, suppose $\beta$ is an assignment of values that satisfies $\phi'$. If $\beta(y)=1$, we have nothing to prove. Thus, assume that $\beta(y)=0$. By Claim \ref{app-claim:correctness}, if $(x,C_1,z_1,C_2,z_2,\ldots,z_k,C_{k+1},y)$ is the path $x\rightarrow y$, then $(\neg x\vee z_1)$, $(\neg z_i\vee z_{i+1})$, $i=1,\ldots,k-1$ and $(\neg z_k\vee y)$ are all clauses of $\phi'$ and are thus satisfied by $\beta$. Since $\beta(y)=0$, we have $\beta(z_k)=0$. Continuing in this way, $\beta(z_i)=0$, $i=1,\ldots,k$ and thus $\beta(x)=0$ too, which implies that $\beta(\neg x \vee y)=1$.  

Now, for the proof of the claim, suppose again that $\phi'$ is satisfiable, and let $\beta$ be an assignment (possibly different than $\alpha$) that satisfies $\phi'$. Since $\beta$ satisfies $\phi'$, it satisfies $\bigvee_{x\in V'} x$. This means that there exists some $x\in V'$ such that $\beta(x)=1$. But $\beta$ also satisfies $(\neg x\vee\neg x')$, so we get that $\beta(x')=0$. Thus $\beta((\neg x \vee x'))=0$, which means that $\phi'$ does not logically entail $\neg x \vee x'$. By the discussion above, there exists no path from $x$ to $x'$, so $x$ is not in a bad \scc \ of $G$.\hfill$\Box$

By Proposition \ref{prop:renpartHorn}, we have seen that $\phi$ is renamable partially Horn if and only if $\phi'$ is satisfiable. Also, in case $\phi'$ is satisfiable, a variable $x\in V$ is renamed if and only if $a(x)=1$.

Thus, by the above and Claim \ref{app-claim:satisf}, $\phi$ is renamable partially Horn if and only if there is some variable $x$ that does not appear in a bad \scc \ of $G$. Furthermore, the process described in order to obtain assignment $a$ is linear in the length of the input formula, and $a$ provides the information about which variables to rename.  
\end{proof}

Because checking whether a formula is affine can be trivially done in linear time, we get: 
\begin{theorem}\label{thm:pic}
There is an algorithm that, on input a formula $\phi$, halts in linear time in the length of $\phi$ and either returns that $\phi$ is not a possibility integrity constraint, or alternatively, (i) either it returns that $\phi$ is affine or (ii) in case $\phi$ is separable, it produces two non-empty and disjoint subsets $V_1,V_2\subseteq V$ such that no clause of $\phi$ contains variables from both $V_1$ and $V_2$ and (iii) in case $\phi$ is  renamable partially Horn, it produces a subset $V^*\subseteq V$ such that the formula  $\phi^*$  obtained from $\phi$ by renaming the literals of variables in $V^*$ is partially Horn.\end{theorem}

We end this section by showing that we can recognize \lpic's efficiently.

\begin{theorem}\label{thm:recupic}
There is an algorithm that, on input a formula $\phi$, halts in linear time in the length of $\phi$ and either returns that $\phi$ is not a local possibility constraint, or alternatively, produces the sets $V_0,V_1,V_2$ described in Definition \ref{def:upic}.\end{theorem}
\begin{proof}
First, we check if $\phi$ is bijunctive or affine (this can be trivially done in linear time). If it is, then $\phi$ is an \lpic. Else, we use the algorithm of Theorem \ref{thm:linearrph} to obtain $V_0$. Note that, by the construction of $G$ and the way we obtain $V_0$, there is no variable in $V\setminus V_0$ that can belong in an admissible set.

If $V_0=\emptyset$, then either $\phi$ is not an \lpic, or there is a partition $(V_1,V_2)$ of $V$ such that no clause of $\phi$ contains variables from both $V_1$ and $V_2$. Thus, we use the algorithm of Proposition \ref{prop:linsep} to check if $\phi$ is separable. If it is not, then $\phi$ is not an \lpic. If it is, we obtain two sub-formulas $\phi_1,\phi_2$ such that $\phi=\phi_1\wedge\phi_2$. We can then trivially check, in linear time to their lengths, if $\phi_1$ and $\phi_2$ are bijunctive and affine respectively, or vice-versa. If they are, then $\phi$ is an \lpic. Else, it is not.

Obviously, if $V_0=V$, then $\phi$ is (renamable) Horn and thus an \lpic. Now, suppose that $(V_0,V\setminus V_0)$ is a partition of $V$. Add all the variables of $V\setminus V_0$ that appear in an $(V,V\setminus V_0)$-generalized clause to $V_2$, and set $V_1=V\setminus(V_0\cup V_2)$. Now, if any clause of $\phi$ contains more that two variables from $V_1$, or variables from both $V_1$ and $V_2$, then $\phi$ is not an \lpic. Else, it is.  
\end{proof}

\begin{remark}\label{problem}
Recall that we have assumed that all the formulas we consider have non-degenerate domains. Note that the above algorithms cannot distinguish such formulas from other formulas of the same form that have degenerate domains. An algorithm that could efficiently decide that, would effectively be (due e.g. to the syntactic form of separable formulas) an algorithm that could decide on input any given formula, which variables are satisfied by exactly one Boolean value and which admit both. It is quite plausible that no such efficient algorithm exists, as it could be used to solve known computationally hard problems, like the \emph{unique satisfiability} problem.
\end{remark}

\section{Syntactic Characterization of (local) possibility domains}\label{sec:character}
In this section, we  provide syntactic characterizationσ for (local) possibility domains, by proving they are the models of (local) possibility integrity constraints. Furthermore, we show that given a (local) possibility domain $D$, we can produce a (local) possibility integrity constraint, whose set of models is $D$, in time polynomial in the size of $D$. To obtain the characterization for possibility domains, we proceed as follows. We separately show that each type of a possibility integrity constraint of Definition \ref{def:posintcon} corresponds to one of the conditions of Theorem \ref{thm:DH}: (i) Domains admitting non-dictatorial binary projection aggregators are the sets of models of separable formulas, those admitting non-projection binary aggregators are the sets of models of renamable partially Horn formulas and (iii) affine domains are the sets of models of affine formulas. For local possibility domains, we directly show they are the models of local possibility integrity constraints. 

We will need some additional notation. For a set of indices $I$, let $D_I:=\{(a_i)_{i\in I}\mid a\in D\}$ be the projection of $D$ to the indices of $I$ and $D_{-I}:=D_{\n\setminus I}$. Also, for two (partial) vectors $a=(a_1,\ldots,a_k)\in D_{\{1,\ldots,k\}}$, $k<n$ and $b=(b_1,\ldots,b_{n-k})\in D_{\{k+1,\ldots,n\}}$, we define their \emph{concatenation} to be the vector $ab=(a_1,\ldots,a_k,b_1,\ldots,b_{n-k}$). Finally, given two subsets $D,D'\subseteq\{0,1\}^n$, we write that $D\approx D'$ if we can obtain $D$ by \emph{permuting} the coordinates of $D'$, i.e. if $D=\{(d_{j_1},\ldots,d_{j_n})\mid (d_1,\ldots,d_n)\in D'\}$, where $\{j_1,\ldots,j_n\}=\n$. 

\subsection{Syntactic characterizations}\label{ssec:syntchar} We begin with characterizing the domains closed under a non-dictatorial projection aggregator as the models of separable formulas.

\begin{proposition}\label{prop:projaggr} $D\subseteq\{0,1\}^n$ admits a binary non-dictatorial projection aggregator $(f_1, \ldots, f_n)$ if and only if there exists a separable formula $\phi$ whose set of models equals $D$.\end{proposition}
We will first need the following lemma:

\begin{lemma}\label{app-lem:cart}
$D$ is closed under a binary non-dictatorial projection aggregator \emph{if and only if} there exists a partition $(I,J)$ of $\n$ such that $D\approx D_I\times D_J$.   
\end{lemma}

\begin{proof}
($\Rightarrow$) Let $F=(f_1,\ldots,f_n)$ be a binary non-dictatorial projection aggregator for $D$. Assume, without loss of generality, that $f_i=\pr_1^2$, $i=1,\ldots,k<n$ and $f_j=\pr_2^2$, $j=k+1,\ldots,n$. Let also $I:=\{1,\ldots,k\}$ and $J:=\{k+1,\ldots,n\}$. Since $k<n$, $(I,J)$ is a partition of $\n$. To prove that $D=D_I\times D_J$, it suffices to prove that $D_I\times D_J\subseteq D$ (the reverse inclusion is always true).

Let $a\in D_I$ and $b\in D_J$. It holds that there exists an $a'\in D_I$ and a $b'\in D_J$ such that both $ab',a'b\in D$. Thus:$$F(ab',a'b)=ab\in D,$$ since $F=(f_1,\ldots,f_n)$ is an aggregator for $D$, $f_i=\pr_1^2$, $i\in I$ and $f_j=\pr_2^2$, $j\in J$.

($\Leftarrow$) Suppose that $D\approx D_I\times D_J$, where $I,J$ is a partition of $\n$. Assume, without loss of generality, that $I=\{1,\ldots,k\}$, $k<n$ and $J=\{k+1,\ldots,n\}$ (thus $D=D_I\times D_J$). Let also $ab',a'b\in D$, where $a,a'\in D_I$ and $b,b'\in D_J$.

Obviously, if $F=(f_1,\ldots,f_n)$ is an $n$-tuple of projections, such that $f_i=\pr_1^2$, $i\in I$ and $f_j=\pr_2^2$, $j\in J$, then $F(ab',a'b)=ab\in D$, since $a\in D_I$ and $b\in D_J$. Thus $F=(f_1,\ldots,f_n)$ is a non-dictatorial projection aggregator for $D$.  
\end{proof}

\begin{proof}[Proof of Proposition \ref{prop:projaggr}]
($\Rightarrow$) Since $D$ admits a binary non-dictatorial projection aggregator $(f_1, \ldots, f_n)$, by Lemma \ref{app-lem:cart}, $D\approx D_I\times D_J$, where $(I,J)$ is a partition of $\n$ such that $I=\{i\mid f_i=\pr_1^2\}$ and $J=\{j\mid f_j=\pr_2^2\}$. Let $\phi_1$ and $\phi_2$ defined on $\{x_i\mid i\in I\}$ and $\{x_j\mid j\in J\}$ respectively, such that $\rMod(\phi_1)=D_I$ and $\rMod(\phi_2)=D_J$. Let also $\phi=\phi_1\wedge\phi_2$. It is straightforward to observe that, since $\phi_1$ and $\phi_2$ contain no common variables:$$\textrm{Mod}(\phi)\approx\textrm{Mod}(\phi_1)\times\textrm{Mod}(\phi_2)=D_I\times D_J\approx D.$$
($\Leftarrow$) Assume that $\phi$ is separable and that $\textrm{Mod}(\phi)=D$. Since $\phi$ is separable, we can find a partition $(I,J)$ of $\n$, a formula $\phi_1$ defined on $\{x_i\mid i\in I\}$ and a $\phi_2$ defined on $\{x_j\mid j\in J\}$, such that $\phi=\phi_1\wedge\phi_2$. Easily, it holds that: $$\textrm{Mod}(\phi)\approx\textrm{Mod}(\phi_1)\times\textrm{Mod}(\phi_2)=D_I\times D_J\approx D.$$ The required now follows by Lemma \ref{app-lem:cart}.
\end{proof}

We now turn our attention to domains closed under binary non projection aggregators. 

\begin{theorem}\label{thm:nonproj}  $D$ admits a binary aggregator $(f_1, \ldots, f_n)$  which is not a projection aggregator if and only if there exists a renamable partially Horn formula $\phi$  whose set of models equals $D$.\end{theorem}

We will first need two lemmas.

\begin{lemma}\label{lem:oneproj}
 Suppose $D$ admits a binary aggregator $F=(f_1,\ldots,f_n)$ such that there exists a partition $(H,I,J)$ of $\n$ where $f_h$ is symmetric for all $h\in H$, $f_i=\pr_s^2$, for all $i\in I$ and $f_j=\pr_t^2$, with $t\neq s$, for all $j\in J$. Then, $D$ also admits a binary aggregator $G=(g_1,\ldots,g_n)$, such that $g_h=f_h$, for all $h\in H$ and $g_i=\pr_s^2$, for all $i\in I\cup J$.\end{lemma}
\begin{proof}
Without loss of generality, assume that there exist $1\leq k <l<n$ such that $H=\{1,\ldots,k\}$, $I=\{k+1,\ldots,l\}$ and $J=\{l+1,\ldots,n\}$ and that $s=1$ (and thus $t=2$). It suffices to prove that, for two arbitrary vectors $a,b\in D$, $G(a,b)\in D,$ where $(g_1,\ldots,g_n)$ is defined as in the statement of the lemma.

Assume that for all $i\in H$, $f_i(a_i,b_i)=c_i$. Since $F$ is an aggregator for $D$, it holds that $F(a,b)$ and $F(b,a)$ are both vectors in $D$. By the same token, so is $F(F(a,b),F(b,a))$. The result is now obtained by noticing that:\begin{align*}
    F(a,b)= & (c_1,\ldots,c_k,a_{k+1},\ldots,a_l,b_{l+1},\ldots,b_n),\\
    F(b,a)= & (c_1,\ldots,c_k,b_{k+1},\ldots,b_l,a_{l+1},\ldots,a_n),
\end{align*}
and thus: $F(F(a,b),F(b,a))=(c_1,\ldots,c_k,a_{k+1},\ldots,a_n)=G(a,b).$
\end{proof}

\begin{lemma}\label{lem:Horn}
Suppose $D$ admits a binary aggregator $(f_1,\ldots,f_n)$ such that, for some $J\subseteq\n$, $f_j$ is symmetric for all $j\in J$. For each $d=(d_1,\ldots,d_n)\in D$, let $d^*=(d_1^*,\ldots,d_n^*)$ be such that: \begin{equation*}
    d_j^*=\begin{cases}
    1-d_j & \text{ if }j\in J,\\
    d_j & \text{ else,}
    \end{cases}
\end{equation*} for $j=1,\ldots,n$ and set $D^*=\{d^*\mid d\in D\}$. Then $D^*$ admits the binary aggregator $(g_1,\ldots,g_n)$, where: (i) $g_j=\wedge$ for all $j\in J$ such that $f_j=\vee$, (ii) $g_j=\vee$ for all $j\in J$ such that $f_j=\wedge$ and (iii) $g_j=f_j$ for the rest.

Furthermore, if there are two formulas $\phi$ and $\phi^*$ such that $\phi^*$ is obtained from $\phi$ by renaming all $x_j$, $j\in J$, then $D=\rMod(\phi)$ if and only if $D^*=\rMod(\phi^*)$.
\end{lemma}
Note that we do not assume that the set $J\subseteq\n$ includes every coordinate $j$ such that $f_j$ is symmetric. 
\begin{proof}
The former statement follows from the fact that $\wedge(1-d_j,1-d'_j)=1-\vee(d_j,d'_j)$ (resp. $\vee(1-d_j,1-d'_j)=1-\wedge(d_j,d'_j)$), for any $d,d'\in D$. For the latter, observe that by renaming $x_j$, $j\in J$, in $\phi$, we cause all of its literals to be satisfied by the opposite value. Thus, $d^*$ satisfies $\phi^*$ if and only if $d$ satisfies $\phi$.  
\end{proof}

For two vectors $a,b\in D$, we define $a\leq b$ to mean that if $a_i=1$ then $b_i=1$, for all $i\in\n$ and $a<b$ when $a\leq b$ and $a\neq b$.
\begin{proof}[Proof of Theorem \ref{thm:nonproj}]
($\Rightarrow$) We will work with the corresponding domain $D^*$ of Lemma \ref{lem:Horn} that admits an aggregator $(g_1,\ldots,g_n)$ whose symmetric components, corresponding to the symmetric components of $(f_1,\ldots,f_n)$, are all equal to $\wedge$. Suppose that $V_0=\{x_i\mid g_i=\wedge\}$. For $D^*$, we compute a formula $\phi=\phi_0\wedge\phi_1$, where $\phi_0$ is defined on the variables of $V_0$ and is Horn and where $\phi_1$ has only negative appearances of variables of $V_0$. The result is then derived by renaming all the variables $x_j$, where $j$ is such that $f_j=\vee$.

Let $I:=\{i\mid f_i\text{ is symmetric}\}$ (by the hypothesis, $I\neq\emptyset$). Let also $J:=\{j\mid f_j=\vee\}$ ($J$ might be empty). Obviously $J\subseteq I$. For each $d=(d_1,\ldots,d_n)\in D$, let $d^*=(d^*_1,\ldots,d^*_n)$, where $d^*_j=1-d_j$ if $j\in J$ and $d^*_i=d_i$ else. Easily, if $D^*=\{d^*\mid d\in D\}$, by Lemma \ref{lem:Horn} it admits an aggregator $(g_1,\ldots,g_n)$ such that $g_i=\wedge$, for all $i\in I$ and $g_j=\pr_1^2$, $j\notin I$. Thus, there is a Horn formula $\phi_0$ on $\{x_i\mid i\in I\}:=V_0$, such that $\textrm{Mod}(\phi_0)=D^*_I$.

If $I=\n$, we have nothing to prove. Thus, suppose, without loss of generality, that $I=\{1,\ldots,k\}$, $k<n$. For each $a=(a_1,\ldots,a_k)\in D^*_I$, let $B_a:=\{b\in D^*_{-I}\mid ab\in D^*\}$ be the set containing all partial vectors that can extend $a$. For each $a\in D^*_{I}$, let $\psi_a$ be a formula on $\{x_j\mid j\notin I\}$, such that $\textrm{Mod}(\psi_a)=B_a$. Finally, let $I_a:=\{i\in I\mid a_i=1\}$ and define:$$\phi_a:=\Bigg(\bigwedge_{i\in I_a}x_i\Bigg)\rightarrow\psi_a,$$ for all $a\in D^*_I$.

Consider the formula:$$\phi=\phi_0\wedge\Bigg(\bigwedge_{a\in D^*_I}\phi_a\Bigg).$$ We will prove that $\phi$ is partially Horn and that $\textrm{Mod}(\phi)=D^*$. By Lemma \ref{lem:Horn}, the renamable partially Horn formula for $D$ can be obtained by renaming in $\phi$ the variables $x_i$ such that $i\in J$.   

We have already argued that $\phi_0$ is Horn. Also, since $\phi_a$ is \emph{logically equivalent} to (has exactly the same models as):$$\Bigg(\bigvee_{i\in I_a}\neg x_i\Bigg)\vee\psi_a,$$ any variable of $V_0$ that appears in the clauses of some $\phi_a$, does so negatively. It follows that $\phi$ is partially Horn.

Next we  show that $D^*\subseteq\textrm{Mod}(\phi)$ and that $\textrm{Mod}(\phi)\subseteq D^*$. For the former inclusion, let $ab\in D^*$, where $a\in D^*_I$ and $b\in B_a$. Then, it holds that $a$ satisfies $\phi_0$ and $b$ satisfies $\psi_a$. Thus $ab$ satisfies $\phi_a$.

Now, let $a'\in D^*_I:a\not\geq a'$. Then, $a$ does not satisfy $\bigwedge_{i\in I_{a'}}x_i$, since there exists some coordinate $i\in I_{a'}$ such that $a_i=0$ and $a'_i=1$. Thus, $ab$ satisfies $\phi_{a'}$. Finally, let $a''\in D^*_I:a''<a$. Then, $a$ satisfies $\bigwedge_{i\in I_{a''}}x_i$ and thus we must prove that $b$ satisfies $\psi_{a''}$. 

Since $a''\in D^*_I$, there exists a $c\in D^*_{-I}$ such that $a''c\in D^*$. Then, since $(g_1,\ldots,g_n)$ is an aggregator for $D^*$:\begin{multline*}
   (g_1,\ldots,g_n)(ab,a''c)=\\(\wedge(a_1,a''_1),\ldots,\wedge(a_k,a''_k),\pr_1^2(b_1,c_1),\ldots,\pr_1^2(b_{n-k},c_{n-k}))=a''b\in D^*, 
\end{multline*}
since $a''<a$. Thus, $b\in B(a'')$ and, consequently, it satisfies $\psi_{a''}$.

We will prove the opposite inclusion by showing that an assignment not in $D^*$ cannot satisfy $\phi$. Let $ab\notin D^*$. If $a\notin D^*_I$, we have nothing to prove, since $a$ does not satisfy $\phi_0$ and thus $ab\notin\textrm{Mod}(\phi)$. So, let $a\in D^*_I$. Then, $b\notin B_a$, lest $ab\in D^*$. But then, $b$ does not satisfy $\psi_a$ and thus $ab$ does not satisfy $\phi_a$. Consequently, $ab\notin\textrm{Mod}(\phi)$.

Thus, by renaming the variables $x_i$, $i\in J$, we produce   a renamable partially Horn formula, call it  $\psi$,  such that $\rMod(\psi)=D$.

($\Leftarrow$) Let $\psi$ be a renamable partially Horn formula with $\textrm{Mod}(\psi)=D$. Let $J\subseteq\n$ such that, by renaming all the $x_i$, $i\in J$, in $\psi$, we obtain a partially Horn formula $\phi$. Let $V_0$ be the set of variables such that any clause containing only variables from $V_0$ is Horn, and that appear only negatively in clauses that contain variables from $V\setminus V_0$. By Remark \ref{rem:ren}, we can assume that $\{x_i\mid i\in J\}\subseteq V_0$. Let also $\Cl_0$ be the set of admissible Horn clauses of $\phi$.

Let again $D^*=\{d^*\mid d\in D\}$, where $d^*_j=1-d_j$ if $j\in J$ and $d^*_i=d_i$ else, for all $d\in D$. By Lemma \ref{lem:Horn}, $\rMod(\phi)=D^*$. By the same Lemma, and by noticing that whichever the choice of $J\subseteq\n$, $(D^*)^*=D$, it suffices to prove that $D^*$ is closed under a binary aggregator $(f_1,\ldots,f_n)$, where $f_i=\wedge$ for all $i$ such that $x_i\in V_0$ and $f_j=\pr_1^2$ for the rest.

Without loss of generality, let $I=\{1,\ldots,k\}$, $k<n$ (lest we have nothing to show) be the set of indices of the variables in $V_0$. We need to show that if $ab,a'b'\in D$, where $a,a'\in D^*_I$ and $b,b'\in D^*_{-I}$, then $(a\wedge a')b\in D$, where $a\wedge a'=(a_1\wedge a'_1,\ldots,a_k\wedge a'_k)$.

Let $\phi=\phi_0\wedge \phi_1$, where $\phi_0$ is the conjunction of the clauses in $\Cl_0$ and $\phi_1$ the conjunction of the rest of the clauses of $\phi$. By the hypothesis, $\phi_0$ is Horn and thus, since $a,a'$ satisfy $\phi_0$, so does $a\wedge a'$. Now, let $C_r$ be a clause of $\phi_1$. If any literal of $C_r$ that corresponds to a variable not in $\phi_0$ is satisfied by $b$, we have nothing to prove. If there is no such literal, since $ab$ satisfies $C_r$, it must hold that a negative literal $\bar{x}_i$, $i\in I$, is satisfied by $a$. Thus, $a_i=0$, which means that $a_i\wedge a'_i=0$ too. Consequently, $C_r$ is satisfied by $(a\wedge a')b$. Since $C_r$ was arbitrary, the proof is complete.  \end{proof}
We thus get:
\begin{theorem}\label{thm:charpd}
$D$ is a possibility domain if and only if there exists a possibility integrity constraint $\phi$ whose set of models equals $D$.
\end{theorem}

\begin{proof}
($\Rightarrow$) If $D$ is a possibility domain, then, by Theorem \ref{thm:DH}, it either admits a non-dictatorial binary projection, or a non-projection binary aggregator or a ternary aggregator all components of which are the binary addition $\mod2$. In the first case, by Proposition \ref{prop:projaggr}, $D$ is the model set of a separable formula. In the second, by Theoremn \ref{thm:nonproj}, it is the model set of a renamable partially Horn formula and in the third, that of an affine formula. Thus, in all cases, $D$ is the model set of a possibility integrity constraint.

($\Leftarrow$) Let $\phi$ be a possibility integrity constraint such that $\rMod(\phi)=D$. If $\phi$ is separable, then, by Proposition \ref{prop:projaggr}, $D$ admits a non-dictatorial binary projection aggregator. If $\phi$ is renamable partially Horn, then, by Theorem \ref{thm:nonproj}, $D$ admits a non-projection binary aggregator. Finally, if $\phi$ is affine, then $D$ admits a ternary aggregator all components of which are the binary addition $\mod2$. In every case, $D$ is a possibility domain.
\end{proof}

We turn now our attention to local possibility domains. Analogously to the case of possibility domains, we characterize \lpd's as the sets of models of \lpic's.
\begin{theorem}\label{thm:updupic}
A domain $D\subseteq\{0,1\}^n$ is a local possibility domain if and only if there is a local possibility integrity constraint $\phi$ such that $\rMod(\phi)=D$. 
\end{theorem}

We will first need two lemmas.
\begin{lemma}\label{lem:partialaggr}
Let $D\subseteq\{0,1\}^n$ and $I=\{j_1,\ldots,j_t\}\subseteq\{1,\ldots,n\}$. Then, if $F=(f_1,\ldots,f_n)$ is a $k$-ary aggregator for $D$, $(f_{j_1},\ldots,f_{j_t})$ is a $k$-ary aggregator for $D_I$.
\end{lemma}
\begin{proof}
Without loss of generality, assume $I=\{1,\ldots,s\}$, where $s\leq n$ and let $a^1,\ldots,a^k\in D_I$. It follows that there exist $b^1,\ldots,b^k\in D_{-I}$ such that $c^1,\ldots,c^k\in D$, where $c^i=a^ib^i$, $i=1,\ldots k$. Since $F$ is an aggregator for $D$: $$F(c^1,\ldots,c^k):=(f_1(c_1^1,\ldots,c_1^k),\ldots,f_n(c_n^1,\ldots c_n^k))\in D.$$ Thus, $(f_1(c_1^1,\ldots,c_1^k),\ldots,f_s(c_s^1,\ldots,c_s^k))\in D_I$.
\end{proof}
\begin{lemma}\label{lem:upd-binary}
Suppose that $D$ admits a ternary aggregator $F=(f_1,\ldots,f_n)$, where $f_j\in\{\wedge^{(3)},\maj,\oplus\}$, $j=1,\ldots,n$. Then $D$ admits a binary aggregator $G=(g_1,\ldots,g_n)$ such that $g_i=\wedge$, for all $i$ such that $f_i=\wedge^{(3)}$, $g_j=\pr_1^2$, for all $j$ such that $f_j=\maj$ and $g_k=\pr_2^2$, for all $k$ such that $f_k=\oplus$.
\end{lemma}
\begin{proof}
The result is immediate, by defining $G=(g_1,\ldots,g_n)$ such that: $$g_j(x,y)=f_j(x,x,y),$$ for $j=1,\ldots,n$.
\end{proof}

\begin{proof}[Proof of Theorem \ref{thm:updupic}]
($\Rightarrow$) The proof will closely follow that of Theorem \ref{thm:nonproj}.

Since $D$ is an \lpd, by Theorem \ref{thm:upd-char}, there is a ternary aggregator $F=(f_1,\ldots,f_n)$ such that every component $f_j\in\{\wedge^{(3)},\vee^{(3)},\maj,\oplus\}$, $j=1,\ldots,n$. Again, let $D^*=\{d^*\mid d\in D\}$, where $d^*_j=1-d_j$ if $j$ is such that $f_j=\vee^{(3)}$, and $d^*_j=d_j$ in any other case. Thus, by Lemma \ref{lem:Horn}, $D^*$ admits a ternary aggregator $G=(g_1,\ldots,g_n)$ such that $g_j\in\{\wedge^{(3)},\maj,\oplus\}$, for $j=1,\ldots,n$. Thus, by showing that $D^*$ is described by a \lpic \ $\phi$, we will obtain the same result for $D$ by renaming all the variables $x_j$, where $j$ is such that $f_j=\vee^{(3)}$.

Without loss of generality, assume that $I:=\{i\mid g_i=\wedge^{(3)}\}=\{1,\ldots,s\}$, $J:=\{j\mid g_j=\maj\}=\{s+1,\ldots,t\}$ and $K:=\{k\mid g_k=\oplus\}=\{t+1,\ldots,n\}$, where $0\leq s\leq t\leq n$. Since $D^*_I$ is Horn, there is a Horn formula $\phi_0$ such that $\rMod(D^*_I)=\phi_0$.

If $s=t=n$, we have nothing to prove. Thus, suppose $s<t\leq n$. For each $a=(a_1,\ldots,a_s)\in D^*_I$, let $B_a^1:=\{b\in D^*_J\mid ab\in D^*_{I\cup J}\}$ and $B^2_a:=\{c\in D^*_J\mid ac\in D^*_{I\cup K}\}$ be the sets of partial vectors extending $a$ to the indices of $J$ and $K$ respectively.

\begin{claim}\label{claim:extending}
For each $a\in D^*_I$, $B^1_a$ and $B^2_a$ are bijunctive and affine respectively.
\end{claim}
\textit{Proof of Claim:} We will prove the claim for $B_a^1$. The proof for $B_a^2$ is the same.

Let $b^1,b^2,b^3\in B_a^1$. Then $ab^1,ab^2,ab^3\in D^*_{I\cup J}$. Since, by Lemma \ref{lem:partialaggr}, $(g_1,\ldots,g_t)$ is an aggregator for $D^*_{I\cup J}$ and by the definition of $G$, it holds that $ab\in D^*_{I\cup J}$, where $b=\maj(b^1,b^2,b^3)$. Thus, $b\in B_a^1$ and the result follows.\hfill$\Box$

Thus, for each $a\in D^*_{I}$, there is a bijunctive formula $\psi_a$ and an affine $\chi_a$, such that $\rMod(\psi_a)=B_a^1$ and $\rMod(\chi_a)=B_a^2$. Let $I_a:=\{i\in I\mid a_i=1\}$ and define:$$\phi^1_a:=\Bigg(\bigwedge_{i\in I_a}x_i\Bigg)\rightarrow\psi_a$$ and $$\phi^2_a:=\Bigg(\bigwedge_{i\in I_a}x_i\Bigg)\rightarrow\chi_a,$$ for all $a\in D^*_I$.

Consider the formula:$$\phi=\phi_0\wedge\Bigg(\bigwedge_{a\in D^*_I}\phi^1_a\Bigg)\wedge\Bigg(\bigwedge_{a\in D_I^*}\phi^2_a\Bigg).$$ Let $V_0=\{x_i\mid i\in I\}$, $V_1=\{x_j\mid j\in J\}$ and $V_2=\{x_k\mid k\in K\}$. That $\phi$ is partially Horn with admissible set $V_0$, can be seen in the same way as in Theorem \ref{thm:nonproj}. Now, consider $\psi_a$, for some $a\in D^*_I$. Since it is bijunctive, it is of the form:$$\psi_a=\bigwedge_{j=1}^r\Bigg(l_{j_1}\vee l_{j_2}\Bigg),$$ where $l_{j_i}$ are literals of variables from $V_1$ ($l_{j_1}$ and $l_{j_2}$ can be the same). Thus, $\phi^1_a$ is equivalent to:$$\bigwedge_{j=1}^r\Bigg(\Bigg(\bigvee_{i\in I_a}\neg x_i\Bigg)\vee l_{j_1}\vee l_{j_2}\Bigg).$$ Thus, the clauses of $\phi^1_a$ contain at most two literals from $V_1$.

In the analogous way, we can see that the clauses of $\phi_a^2$ are $(V_0,V_1)$-generalized clauses. Finally, by construction, there is no clause in $\phi$ that contains variables both from $V_1$ and $V_2$. It follows that $\phi$ is an \lpic. What remains now is to show that $\rMod(\phi)=D^*$. 

By Lemmas \ref{lem:oneproj} and \ref{lem:upd-binary}, it follows that $D^*$ admits a binary aggregator $H=(h_1,\ldots,h_n)$ such that $h_i=\wedge$, for all $i\in I$ and $h_j=\pr_1^2$, for all $j\in J\cup K$. The proof now is exactly like the one of Theorem \ref{thm:nonproj}, by letting $B_a=\{bc\mid b\in B_a^1\text{ and }c\in B_a^2\}$ and $$\phi_a=\phi_a^1\wedge\phi_a^2.$$

($\Leftarrow$) Let $\psi$ be an \lpic, with $\textrm{Mod}(\psi)=D$. Let $V_0, V_1$ and $V_2$ be subsets of $V$ as in Definition \ref{def:upic}. Let also $\phi$ be the partially Horn formula obtained by $\psi$ by renaming the variables of a subset $V^*\subseteq V_0$. Again, assume $D^*=\{d^*\mid d\in D\}$, where $d^*_j=1-d_j$ if $x_j\in V^*$ and $d^*_i=d_i$ else, for all $d\in D$. By Lemma \ref{lem:Horn}, $\rMod(\phi)=D^*$. Thus, by Theorem \ref{thm:upd-char}, it suffices to prove that $D^*$ is closed under a ternary aggregator $(f_1,\ldots,f_n)$, where $f_i\in\{\wedge^{(3)},\maj,\oplus\}$ for $i=1,\ldots,n$.

Without loss of generality, let $I=\{1,\ldots,s\}$, be the set of indices of the variables in $V_0$, $J=\{s+1\ldots,t\}$ be that of the indices of variables in $V_1$ and $K=\{t+1,\ldots,n\}$ that of the indices of variables in $V_2$. We need to show that if $abc,a'b'c',a''b''c''\in D^*$, where $a,a',a''\in D^*_I$, $b,b',b''\in D^*_J$ and $c,c',c''\in D^*_K$, then $$d:=(\wedge^{(3)}(a,a',a''),\maj(b,b',b''),\oplus(c,c',c''))\in D^*.$$

Let $\phi=\phi_0\wedge \phi_1\wedge\phi_2$, where $\phi_0$ is the conjunction of the clauses containing only variables from $V_0$, $\phi_1$ the conjunction of clauses containing variables from $V_1$ and where $\phi_2$ contains the rest of the clauses of $\phi$. Observe that by the hypothesis, there is no variable appearing both in a clause of $\phi_1$ and $\phi_2$.

By the hypothesis, $\phi_0$ is Horn and thus, since $a,a',a''$ satisfy $\phi_0$, so does $a\wedge a'\wedge a''$.

Now, let $C_r$ be a clause of $\phi_1$. Suppose that there is a literal of a variable $x_i\in V_0$ in $C_r$ that is satisfied by $a$. Since $\phi$ is partially Horn with respect to $V_0$, it must hold that this literal was $\neg x_i$. This means that $a_i=0$ and thus $\wedge^{(3)}(a_i,a'_i,a''_i)=0$. The same holds if $\neg x_i$ is satisfied by $a'$ or $a''$. Thus, $C_r$ is satisfied.  

Now, suppose there is no such literal and that the literals of $C_r$ corresponding to variables of $V_1$ are $l_i$, $l_j$ (again $l_i$ and $l_j$ need not be different). Since $abc,a'b'c',a''b''c''$ satisfy $\phi$, it holds that $(b_i,b_j)$, $(b_i',b_j')$ and $(b''_i,b''_j)$ satisfy $l_i\vee l_j$. Without loss of generality, Assume that $\maj(b_i,b_i',b_i'')=b_i$ and that $b_i$ does not satisfy $l_i$, lest we have nothing to prove (this cannot be the case if $l_i=l_j$). Then, $b_i=b_i'$ or $b_i=b_i''$. Assume the former (again without loss of generality). Then, it must be the case that $b_j,b_j'$ satisfy $l_j$. Thus $b_j=b_j'$ and $\maj(b_j,b_j',b_j'')=b_j$, which satisfies $l_j$. In every case, $C_r$ is satisfied by $d$.

Now, let $C_q$ be a clause of $\phi_2$. Again, if there there is a literal of a variable $x_i\in V_0$ in $C_q$ that is satisfied by $a$, we obtain the required result as in the case of $C_r$. Thus, suppose there is no such literal and that the sub-clause of $C_q$ obtained by deleting the variables of $V_0$ is:$$C'_q=(l_1\oplus\cdots\oplus l_z).$$
Since $abc,a'b'c',a''b''c''$ satisfy $\phi$, it holds that $c,c'$ and $c''$ satisfy $C_q'$. Since $C_q'$ is affine, it holds that $\oplus(c,c',c'')$, and satisfies it.

In all cases, we proved that $d$ satisfies $\phi$ and thus the proof is complete.
\end{proof}

\subsection{Efficient constructions}\label{ssec:efficient} To finish this section, we will use Zanuttini and H\'ebrard's ``unified framework'' \cite{zanuttini2002unified}. Recall the definition of a prime formula (Def. \ref{def:prime}) and consider the following proposition:

\begin{proposition}\label{lemma:prime}
Let $\phi_P$ be a prime formula and $\phi$ be a formula logically equivalent to $\phi_P$. Then:\begin{enumerate}
    \item if $\phi$ is separable, $\phi_P$ is also separable and
    \item if $\phi$ is renamable partially Horn, $\phi_P$ is also renamable partially Horn.
\end{enumerate}\end{proposition}
\begin{proof} Let $\phi_P$ be a prime formula. Quine \cite{quine1959cores} showed that the prime implicates of  $\phi_P$ can be obtained from any formula $\phi$ logically equivalent to $\phi_P$, by repeated (i) resolution and (ii) omission of the clauses that have sub-clauses already created. Thus, using the procedures (i) and (ii) on $\phi$, we can obtain every clause of $\phi_P$.

If $\phi$ is separable, where $(V',V\setminus V')$ is the partition of its vertex set such that no clause contains variables from both $V'$ and $V\setminus V'$, it is obvious that neither resolution or omission can create a clause that destroys that property. Thus, $\phi_P$ is separable. 

Now, let $\phi$ be a renamable partially Horn formula where, by renaming the variables of $V^*\subseteq V$, we obtain the partially Horn formula $\phi^*$, whose admissible set of variables is $V_0$. Let also $\phi_P^*$ be the formula obtained by renaming the variables of $V^*$ in $\phi_P$. Easily, $\phi_P^*$ is prime.

Observe that the prime implicates of a partially Horn formula, are also partially Horn. Indeed, it is not difficult to observe that neither resolution, nor omission can cause a variable to seize being admissible: suppose $x\in V_0$. Then, the only way that it can appear in an inadmissible set due to resolution is if there is an admissible Horn clause $C$ containing $\neg x, y$, where $y\in V_0$ too and an inadmissible clause $C'$ containing $\neg y$. But then, after using resolution, $x$ appears negatively to the newly obtained clause. Thus, $\phi_P^*$ is partially Horn, which means that $\phi_P$ is renamable partially Horn.\end{proof}

We are now ready to prove our first main result: 
\begin{theorem}\label{thm:main}
There is an algorithm that, on input $D\subseteq\{0,1\}^n$, halts  in time $O(|D|^2n^2)$ and either returns that $D$ is not a possibility domain, or alternatively outputs a possibility integrity constraint $\phi$, containing $O(|D|n)$ clauses, whose set of satisfying truth assignments is $D$.
\end{theorem}
\begin{proof}
Given a domain $D$, we first use Zanuttini and H\'ebrard's algorithm to check if it is affine \cite[Proposition 8]{zanuttini2002unified}, and if it is, produce, in time $O(|D|^2n^2)$ an affine formula $\phi$ with $O(|D|n)$ clauses, such that $\rMod(\phi)=D$. If it isn't, we use again Zanuttini and H\'ebrard's algorithm \cite{zanuttini2002unified} to produce, in time $O(|D|^2n^2)$, a prime formula $\phi$ with $O(|D|n)$ clauses, such that $\rMod(\phi)=D$. Then, we use the linear algorithms of Proposition \ref{prop:linsep} and Theorem \ref{thm:linearrph} to check if $\phi$ is separable or renamable partially Horn. If it is either of the two, then $\phi$ is a possibility integrity constraint and, by Theorem \ref{thm:charpd}, $D$ is a possibility domain. Else, by Proposition \ref{lemma:prime}, $D$ is not a possibility domain.
\end{proof}

We end this section by proving our second main result, that given an \lpd \ $D$, we can efficiently construct an \lpic \ $\phi$ such that $\rMod(\phi)=D$.
\begin{theorem}\label{thm:main2}
There is an algorithm that, on input $D\subseteq\{0,1\}^n$, halts in time $O(|D|^2n^2)$ and either returns that $D$ is not a local possibility domain, or alternatively outputs a local possibility integrity constraint $\phi$, containing $O(|D|n)$ clauses, whose set of satisfying truth assignments is $D$.
\end{theorem}

We first briefly discuss some results of Zanuttini and H\'ebrard. By \cite[Proposition 3]{zanuttini2002unified}, we get that a prime formula that is logically equivalent to a bijunctive one, is also bijunctive. Now, for a clause $C=l_1\vee\ldots\vee l_t$, where $l_j$ are literals, $j=1,\ldots,t$, let $E(C)=l_1\oplus\ldots\oplus l_t$. For a CNF formula $\phi=\bigwedge_{j=1}^m C_j$, let $A(\phi)=\bigwedge_{j=1}^m E(C)$. In \cite[Proposition 8]{zanuttini2002unified}, it is proven that if $\phi$ is prime, $\rMod(\phi)=D$ and $D$ is affine, then $\rMod(A(\phi))=D$.

\begin{proof}[Proof of Theorem \ref{thm:main2}]
Given a domain $D$, we first use Zanuttini and H\'ebrard's algorithm \cite{zanuttini2002unified} to produce, in time $O(|D|^2n^2)$, a prime formula $\phi$ with $O(|D|n)$ clauses, such that $\rMod(\phi)=D$. Note that at this point, $\phi$ does not contain any generalized clauses (see below). We then use the linear algorithm of Theorem \ref{thm:linearrph} to produce a set $V_0$ such that $\phi$ is renamable partially Horn with admissible set $V_0$.

If $V_0=V$ we have nothing to prove. Thus, suppose that $\phi=\phi_0\wedge\phi_1$, where $\phi_0$ contains only variables from $V_0$. Let $\phi'_1$ be the sub-formula of $\phi_1$, obtained by deleting all variables of $V_0$ from $\phi$. We use the algorithm of Proposition \ref{prop:linsep} to check if $\phi'_1$ is separable.

Suppose that $\phi'_1$ is not separable. We then check, with Zanuttini and H\'ebrard's algorithm, if $\phi'_1$ is either bijunctive or affine. If it is neither, then $D$ is not an \lpd. If it is bijunctive, then $\phi$ is a $\lpic$. If it is affine, we construct the formula $A^*(\phi_1)$ as follows. For each clause $C=(l_1 \vee \cdots \vee l_s\vee(l_{s+1}\vee\cdots\vee l_t)),$ where $l_1,\ldots,l_s$ are literals of variables in $V_0$, let: $$E^*(C)=(l_1 \vee \cdots \vee l_s\vee(l_{s+1}\oplus\cdots\oplus l_t))$$ and $A^*(\phi_1)=\bigwedge_{j=1}^m E^*(C_j)$. Then, the \lpic \ that describes $D$ is $\phi_0\wedge A^*(\phi_1)$.  

In case $\phi'_1$ is separable, assume that $\phi'_!=\phi'_2\wedge\phi'_3$, where no variable appears in both $\phi'_2$ and $\phi'_3$. Let also $\phi_2$ be $\phi'_2$ with the variables from $V_0$ and respectively for $\phi_3$. We now proceed exactly as with $\phi'_1$, but separately for $\phi'_2$ and $\phi'_3$. If either one is neither bijunctive nor affine, $D$ is not an \lpd. Else, we produce the corresponding \lpic \ as above.
\end{proof}

\section{Other Forms of non-Dictatorial Aggregation}\label{sec:other}
In this section, we discuss four different notions of non-dictatorial aggregation procedures that have been introduced in the field of judgment aggregation: aggregators that are not generalized dictatorships, and anonymous, monotone and StrongDem aggregators. We prove that pic's, lpic’s,  a sub-class of pic’s,  and a subclass of lpic’s, respectively,  describe domains that admit each of the above four kind of aggregators. Then, we consider the property of systematicity and examine how our results change if the aggregators are required to satisfy it.
\subsection{Generalized Dictatorships}\label{ssec:gendict}
We begin by defining generalized dictatorships.
\begin{definition}\label{def:gendict}
Let $F =(f_1,\ldots,f_n)$ be an $n$-tuple of $k$-ary conservative functions. $F$ is a \emph{generalized dictatorship for a domain} $D\subseteq\{0,1\}^n$, if, for any $x^1,\ldots,x^k\in D$, it holds that:
\begin{equation}\label{gendic}F(x^1,\ldots,x^k):=(f_1(x_1),\ldots,f_n(x_n))\in\{x^1,\ldots,x^k\}.\end{equation}
\end{definition}
Much like dictatorial functions, it is straightforward to observe that if $F$ is a generalized dictator for $D$, then it is also an aggregator for $D$.

It should be noted here that in the original definition of Grandi and Endriss \cite{grandi2013lifting}, generalized dictatorships are defined independently of a specific domain. Specifically, condition \eqref{gendic} is required to hold for all $x^1,\ldots,x^k\in\{0,1\}^n$. With this stronger definition, they show that the class of generalized dictators coincides with that of functions that are aggregators for every domain $D\subseteq\{0,1\}^n$.

\begin{remark}\label{remark:gendict}
The difference in the definition of generalized dictatorships comes from a difference in the framework we use. Here, we opt to consider the aggregators \emph{restricted} in the given domain, in the sense that we are not interested in what they do on inputs that are not allowed by it. The implications of this are not very evident in the Boolean framework, especially since we consider aggregators that satisfy IIA, on non-degenerate domains (in fact, this issue will not arise in any other aggregator present in this work, apart from generalized dictatorships). On the other hand, in the non-Boolean framework, using unrestricted aggregators could result in trivial cases of non-dictatorial aggregation, where the aggregator is not a projection only on inputs that are not allowed by the domain.  
\end{remark}

The following example shows that the result of Grandi and Endriss \cite[Theorem $16$]{grandi2013lifting} does not hold in our setting.

\begin{example}\label{ex:gendic}
Consider the Horn formula: $$\phi_{11}=(x_1\vee\neg x_2\vee\neg x_3)\wedge(\neg x_1\vee x_2\vee\neg x_3)\wedge(\neg x_1\vee\neg x_2\vee x_3)\wedge(\neg x_1\vee\neg x_2\vee\neg x_3),$$ whose set of satisfying assignments is:
$$\rMod(\phi_{11})=\{(0,0,0),(0,0,1),(0,1,0),(1,0,0)\}.$$
By definition, $\rMod(\phi_{11})$ is a Horn domain and it thus admits the binary symmetric aggregator $\bar{\wedge}=(\wedge,\wedge,\wedge)$. Furthermore, $\bar{\wedge}$ is not a generalized dictatorship for $\rMod(\phi_{11})$, since $\bar{\wedge}((0,0,1),(0,1,0))=(0,0,0)\notin\{(0,0,1),(0,1,0)\}$.

On the other hand, consider the Horn formula:
$$\phi_{12}=(\neg x_1\vee x_2)\wedge(x_2\vee\neg x_3)\wedge(\neg x_1 \vee\neg x_2\vee x_3).$$ $\bar{\wedge}$ is again an aggregator for the Horn domain: $$\rMod(\phi_{12})=\{(0,0,0),(0,1,0),(0,1,1),(1,1,1)\},$$ but, contrary to the previous case, $\bar{\wedge}$ is a generalized dictatorship for $\rMod(\phi_{12})$, since it is easy to verify that for any $x,y\in D'$, $\bar{\wedge}(x,y)\in\{x,y\}$.

Finally, observe that $(\wedge,\vee,\vee)$ is an aggregator for $\rMod(\phi_{12})$ that is not a generalized dictatorship. The latter claim follows from the fact that: $$(\wedge,\vee,\vee)((0,1,0),(1,1,1))=(0,1,1)\notin\{(0,1,0),(1,1,1)\},$$ while the former is left to the reader. Thus, interestingly enough, $\phi_{12}$ describes a domain admitting an aggregator that is not a generalized dictatorship, although it is not the aggregator that ``corresponds'' to the formula.\hfill$\diamond$ 
\end{example}

It is easy to see that $(\pr_i^k,\ldots,\pr_i^k)$ is a generalized dictatorship of any $D\subseteq\{0,1\}^n$, for all $k\geq 1$ and for all $i\in\{1,\ldots,k\}$. Thus, trivially, every domain admits aggregators which are generalized dictatorships. On the other hand, every domain $D\subseteq\{0,1\}^n$ containing only two elements (a domain cannot contain less than two due to non-degeneracy) admits only generalized dictatorships. Indeed, assume $D=\{x,y\}$, $x\neq y$ and let $F$ be a $k$-ary aggregator for $D$. Obviously, $F(x,\ldots,x)=x$ and $F(y,\ldots,y)=y$, since $F$ is unanimous. Also, $F(x,y)\in\{x,y\}=D$ since $F$ is an aggregator.

Our aim is again to find a syntactic characterization for domains that admit aggregators which are not generalized dictatorships. The following result shows that these domains are all the possibility domains with at least three elements, and are thus characterized by possibility integrity constraints.

\begin{theorem}\label{thm:gendict}
A domain $D\subseteq\{0,1\}^n$, with at least three elements, admits an aggregator that is not a generalized dictatorship if and only if it is a possibility domain.
\end{theorem}
\begin{proof}
The forward direction is obtained by the trivial fact that an aggregator that is not a generalized dictatorship is also non-dictatorial.

Now, suppose that $D$ is a possibility domain. Then it is either affine or it admits a binary non-dictatorial aggregator. We begin with the affine case. It is a known result that $D\subseteq\{0,1\}^n$ is affine if and only if it is closed under $\oplus$, or, equivalently, if it admits the minority aggregator: $$\boplus=(\underbrace{\oplus,\ldots,\oplus}_{n\text{-times}})$$

\begin{claim}
Let $D\subseteq\{0,1\}^n$ be an affine domain. Then, the \emph{minority} aggregator: $$\boplus=(\underbrace{\oplus,\ldots,\oplus}_{n\text{-times}})$$ is not a generalized dictatorship for $D$.    
\end{claim}
\begin{proof}
 Let $x,y,z\in D$ be three pairwise distinct vectors. Since $y\neq z$, there exists a $j\in\{1,\ldots,n\}$ such that $y_j\neq z_j$. It follows that $y_j+z_j\equiv 1(\mod 2)$. This means that $\oplus(x_j,y_j,z_j)\neq x_j$ and thus that $\boplus(x,y,z)\neq x$. In the same way we show that $\bar{\oplus}(x,y,z)\notin\{x,y,z\}$, which is a contradiction, since $\bar{\oplus}$ is an aggregator for $D$. 
\end{proof}

Recall that the only binary unanimous functions are $\wedge,\vee,\pr_1^2,\pr_2^2$.

\begin{claim}\label{lem:non-symmetric}
Suppose $D\subseteq\{0,1\}^n$ admits a binary non-dictatorial non-symmetric aggregator $F=(f_1,\ldots,f_n)$. Then $F$ is not a generalized dictatorship.
\end{claim}

\begin{proof}
Assume, to obtain a contradiction, that $F$ is a generalized dictatorship for $D$ and let $x,y\in D$. Then, $F(x,y):=z\in\{x,y\}$. Assume that $z=x$. The case where $z=y$ is analogous.

Let $J\subseteq\n$ such that $f_j$ is symmetric, for all $j\in J$ and $f_j$ is a projection otherwise. Note that $J\neq\n$. Let also $I\subseteq\n\setminus J$, such that $f_i=\pr_2^2$, for all $i\in I$ and $f_i=\pr_1^2$ otherwise. If $I\neq\emptyset$, then, for all $i\in I$, it holds that: $$y_i=pr_2^2(x_i,y_i)=f_i(x_i,y_i)=z_i=x_i.$$ Since $x,y$ were arbitrary, it follows that $D_i=\{x_i\}$, for all $i\in I$. Contradiction, since $D$ is non-degenerate.

If $I=\emptyset$, then $f_j=\pr_1^2$, for all $j\notin J$. Note that in that case, $J\neq\emptyset$, lest $F$ is dictatorial. Now, consider $F(y,x):=w\in\{x,y\}$ since $F$ is a generalized dictatorship. By the definition of $F$, $w_j=z_j=x_j$, for all $j\in J$, and $w_i=y_i$, for all $i\notin J$. Thus, if $w=x$, $D$ is degenerate on $\n\setminus J$, whereas if $w=y$, $D$ is degenerate on $J$. In both cases, we obtain a contradiction.
\end{proof}

The only case left is when $D\subseteq\{0,1\}^n$ admits a binary symmetric aggregator. Contrary to the previous case, where we showed that the respective non-dictatorial aggregators could not be generalized dictatorships, here we cannot argue this way, as Example \ref{ex:gendic} attests. Interestingly enough, we show that as in Example \ref{ex:gendic}, we can always find some symmetric aggregator for such a domain that is not a generalized dictatorship.

\begin{claim}\label{lem:symmetric}
Suppose $D\subseteq\{0,1\}^n$ admits a binary non-dictatorial symmetric aggregator $F=(f_1,\ldots,f_n)$. Then, there is a binary symmetric aggregator $G=(g_1,\ldots,g_n)$ for $D$ ($G$ can be different from $F$) that is not a generalized dictatorship for $D$.
\end{claim}
\begin{proof}
If $F$ is not a generalized dictatorship for $D$, we have nothing to prove. Suppose it is and let $J\subseteq\n$, such that $f_j=\vee$, for all $j\in J$ and $f_i=\wedge$ for all $i\notin J$ ($J$ can be both empty or $\n$).

Let $D^*=\{d^*=(d_1^*,\ldots,d_n^*)\mid d=(d_1,\ldots,d_n)\in D\}$, where: $$d_j^*=\begin{cases}1-d_j &\mbox{ if }j\in J\\
d_j&\mbox{ else.}\end{cases}$$

By Lemma \ref{lem:Horn}, $H=(h_1,\ldots,h_n)$ is a symmetric aggregator for $D$ if and only if $H^*=(h_1^*,\ldots,h_n^*)$ is an aggregator for $D^*$, where $h_j^*=h_j$, for all $j\notin J$ and, for all $j\in J$, if $h_j=\vee$, then $h_j^*=\wedge$ and vice-versa. As expected, the property of being a generalized dictatorship carries on this transformation.

\begin{claim}\label{claim:DtoD*}
$H$ is a generalized dictatorship for $D$ if and only if $H^*$ is a generalized dictatorship for $D^*$.
\end{claim}
\begin{proof}
Let $x=(x_1,\ldots,x_n)$, $y=(y_1,\ldots,y_n)\in D$ and $z:=H(x,y)$. Since $\vee(x_j,y_j)=1-\wedge(1-x_j,1-y_j)$ and $\wedge(x_j,y_j)=1-\vee(1-x_j,1-y_j)$, it holds that $z_j=h^*_j(x_j^*,y_j^*)$, for all $j\notin J$, and $1-z_j=h_j^*(x^*_j,y^*_j)$, for all $j\in J$. Thus, $z^*=H^*(x^*,y^*)$. It follows that $z\in\{x,y\}$ if and only if $z^*\in\{x^*,y^*\}$.
\end{proof}

Now, since $D$ admits the generalized dictatorship $F$, it follows that $D^*$ admits the binary aggregator $\bar{\wedge}=(\underbrace{\wedge,\ldots,\wedge)}_{n\text{-times}}$, that is also a generalized dictatorship. Our aim is to show that $D^*$ admits a symmetric aggregator that is not a generalized dictatorship. The result will then follow by Claim \ref{claim:DtoD*}.

For two elements $x^*,y^*\in D^*$, we write $x^*\leq y^*$ if, for all $j\in\n$ such that $x^*_j=1$, it holds that $y^*_j=1$.

\begin{claim}\label{claim:order}
$\leq$ is a \emph{total ordering} for $D^*$.
\end{claim}
\begin{proof}
To obtain a contradiction, let $x^*,y^*\in D^*$ such that neither $x^*\leq y^*$ nor $y^*\leq x^*$. Thus, there exist $i,j\in\n$, such that $x_i^*=1$, $y_i^*=0$, $x_j^*=0$ and $y_j^*=1$. Thus: 
$$\wedge(x_i^*,y_i^*)=\wedge(x_j^*,y_j^*)=0.$$ Then, $\bar{\wedge}(x^*,y^*)\notin\{x^*,y^*\}$. Contradiction, since $\bar{\wedge}$ is a generalized dictatorship.
\end{proof}
Thus, we can write $D^*=\{d^1,\ldots,d^N\}$, where $d^s\leq d^t$ if and only if $s\leq t$. Let $I\subseteq\n$ be such that, for all $j\in I$: $d_j^s=0$ for $s=1,\ldots,N-1$, and $d_j^N=1$. Observe that $I$ cannot be empty, lest $d^N=d^{N-1}$ and that $I\neq\n$, since $|D|\geq 3$. Let now $G=(g_1,\ldots,g_n)$ such that $g_j=\wedge$, for all $j\in I$ and $g_j=\vee$, for all $j\notin I$.

$G$ is an aggregator for $D^*$. Indeed, let $d^s,d^t\in D^*$ with $s\leq t\leq N-1$. Then, for all $j\notin I$: $$g_j(d_j^s,d_j^t)=\vee(d_j^s,d_j^t)=d_j^t.$$ Also for all $j\in I$: $$g_j(d_j^s,d_j^t)=\wedge(d_j^s,d_j^t)=0=d_j^t.$$ Thus, $G(d^s,d^t)=d^t\in D^*$. Finally, consider $G(d^s,d^N)$. Again, $g_j(d_j^s,d_j^N)=\wedge(d_j^s,d_j^N)=0$ for all $j\in I$ and $g_j(d_j^s,d_j^N)=\vee(d_j^s,d_j^N)=d_j^N$, for all $j\notin I$. By definition of $I$, $G(d^s,d^N)=d^{N-1}\in D^*$. This, last point shows also that $G$ is not a generalized dictatorship, since, for any $s\neq N-1$, $d^{N-1}\notin \{d^s,d^N\}$.
\end{proof}
This completes the proof of Theorem \ref{thm:gendict}.
\end{proof}
By Theorems \ref{thm:charpd} and \ref{thm:gendict} we obtain the following result. 

\begin{corollary}\label{cor:gendict}
A domain $D\subseteq\{0,1\}^n$, with at least three elements, admits an aggregator that is not a generalized dictatorship if and only if there exists a possibility integrity constraint whose set of models equals $D$.
\end{corollary}

\begin{remark}\label{problem:gendict}
What about knowing if a possibility integrity constraint really describes a domain that admits an aggregator that is not a generalized dictatorship? For the requirement of having a non-degenerate domain, the situation is the same as in Remark \ref{problem}. For the requirement of the domain having at least three elements, given that it is non-degenerate, it is easy to see that such domains can only arise as the truth sets of possibility integrity constraints that are Horn, renamable Horn or affine. In all these cases, Creignou and H\'ebrard \cite{creignou1997generating} have devised \emph{polynomial-delay} algorithms that generate all the solutions of such formulas, which can easily be implemented to terminate if they find more than two solutions.    
\end{remark}

\subsection{Preliminaries for Anonymous, Monotone and StrongDem Aggregators}\label{ssec:threetypes}
Our final results concern three kinds of non-dictatorial aggregators, whose properties are based on the majority aggregator.

\begin{definition}\label{def:types}
Let $D\subseteq\{0,1\}^n$. A $k$-ary aggregator $F=(f_1,\ldots,f_n)$ for $D$ is:\begin{enumerate}
    \item Anonymous, if it holds that for all $j\in\n$ and for any permutation $p:\{1,\ldots,k\}\mapsto\{1,\ldots,k\}$:$$f_j(a_1,\ldots,a_k)=f_j(a_{p(1)},\ldots,a_{p(k)}),$$ for all $a_1,\ldots,a_k\in\{0,1\}$.
    \item Monotone, if it holds that for all $j\in\n$ and for all $i\in\{1,\ldots,k\}$: $$f_j(a_1,\ldots,a_{i-1},0,a_{i+1},\ldots,a_k)=1\Rightarrow f_j(a_1,\ldots,a_{i-1},1,a_{i+1},\ldots,a_k)=1.$$
    \item StrongDem, if it holds that for all $j\in\n$ and for all $i\in\{1,\ldots,k\}$, there exist $a_1,\ldots,a_{i-1},a_{i+1},\ldots,a_k\in\{0,1\}$: $$f_j(a_1,\ldots,a_{i-1},0,a_{i+1},\ldots,a_k)=f_j(a_1,\ldots,a_{i-1},1,a_{i+1},\ldots,a_k).$$ 
\end{enumerate}
\end{definition}
Anonymous aggregators ensure that all the voters are treated equally, while monotone that if more voters agree with the aggregator's result, then the outcome does not change. From a Social Theoretic point of view, Nehring and Puppe \cite{nehring2010abstract} have argued that ``For Arrowian (i.e. independent) aggregators, monotonicity is extremely natural, and it is hard to see how non-monotone Arrowian aggregators could be of interest in practice.'' StrongDem aggregators were introudced by Szegedy and Xu \cite{szegedy2015impossibility}. The idea here is that there is a way to fix the votes of any $k-1$ voters such that the remaining voter cannot change the outcome of the aggregation. Apart from the interest these aggregators have for Judgement Aggregation, Szegedy and Xu show that they have strong algebraic properties, as they relate to a property of functions called \emph{strong resilience} (see again \cite{kun2016new,szegedy2015impossibility}). 

Notationally, since these properties are defined for each component of an aggregator, we will say that a Boolean function $f$ is anonymous or monotone if it satisfies property $1$ or $2$ of Definition \ref{def:types} respectively. A Boolean function $f$ that satisfies property $3$ of Definition \ref{def:types} has appeared in the bibliography under the name of $1$-\emph{immune} (see \cite{kun2016new}).  
The first immediate consequence of Definition \ref{def:types}, is that an anonymous or StrongDem aggregator is non-dictatorial. On the other hand, projection and binary symmetric functions are easily monotone, thus every dictatorial and every binary aggregator is monotone. Furthermore, since projections are neither anonymous nor $1$-immune and by Theorem \ref{thm:DH} it is straightforward to observe the following results.
\begin{corollary}\label{cor:firstimplications}
Any possibility domain $D$ either admits a monotone non-dictatorial aggregator or an anonymous one. Furthermore, a domain $D$ admitting an anonymous or StrongDem aggregator is a local possibility domain.
\end{corollary}

Regardless of Corollary \ref{cor:firstimplications}, we can find non-dictatorial aggregators that are neither anonymous, nor monotone, nor StrongDem. An easy such example is an aggregator with at least one component being $\pr_1^3$ and another being $\oplus$, since $\pr_1^3$ is not anonymous,  $\oplus$ is not monotone and neither of the two is 1-immune. Corollary \ref{cor:firstimplications} implies that a domain admitting such an aggregator, admits also another that is monotone or anonymous.

We proceed now with some examples that highlight the various connections between these types of aggregators. 
\begin{example}\label{ex:types}
Any renamable Horn or bijunctive formula describes a domain admitting a symmetric or majority aggregator respectively. Such aggregators are anonymous, monotone and StrongDem. For a more complicated example, consider the formula $$\phi_{13}=(\neg x_1\vee x_2)\wedge(x_2\vee x_3\vee x_4),$$ whose set of satisfying assignments is the local possibility domain:$$\rMod(\phi_{13})=\{0,1\}^4\setminus\Big((\{(1,0)\}\times\{0,1\}^2)\cup\{(0,0,0,0)\}\Big).$$
It is straightforward to check that $\rMod(\phi_{13})$ admits the anonymous, monotone and StrongDem aggregator $(\wedge^{(3)},\vee^{(3)},\maj,\maj)$.

On the other hand, consider the affine formula $$\phi_{14}=x_1\oplus x_2\oplus x_3,$$ where: $$\rMod(\phi_{14})=\{(0,0,1),(0,1,0),(1,0,0),(1,1,1)\}.$$ It can be proven (by a combination of results by Dokow and Holzman \cite[Example $3$]{dokow2010aggregation} and Kirousis et al. \cite[Example $4.5$]{kirousis2017aggregation}) that $\rMod(\phi_{14})$ does not admit any monotone or StrongDem aggregators. On the other hand, it does admit the anonymous aggregator $\bar{\oplus}=(\oplus,\oplus,\oplus)$.

Recall that in Example \ref{ex:possdom}, we argued that the set of satisfying assignments of the formula: $\phi_7=(\neg x_1\vee x_2 \vee x_3)\wedge(x_1\vee\neg x_2\vee\neg x_3)$ is the impossibility domain $\rMod(\phi_7)=\{0,1\}^3\setminus\{(1,0,0),(0,1,1)\}.$
Consider also, from the same example, the formula $\phi_9=(\neg x_1\vee x_2 \vee x_3)\wedge(x_1\vee\neg x_2\vee\neg x_3)\wedge(\neg x_4\vee x_5 \vee x_6)\wedge(x_4\vee\neg x_5\vee\neg x_6),$ whose set of satisfying assignments is:
$\rMod(\phi_9)=\rMod(\phi_7)\times\rMod(\phi_7).$
$\rMod(\phi_9)$ is a possibility domain admitting the monotone aggregator $(\pr_1^2,\pr_1^2,\pr_1^2,\pr_2^2,\pr_2^2,\pr_2^2).$ On the other hand, $\rMod(\phi_9)$ admits neither anonymous, nor StrongDem aggregators, as it is not a local possibility domain.\hfill$\diamond$\end{example}

We now provide examples of StrongDem aggregators that are either not anonymous or not monotone. Domains admitting such aggregators can be proven to also admit aggregators that are anonymous and monotone (see Theorem \ref{thm:monotone,StrongDem} below). 
\begin{example}\label{ex:aggrtypes}
Let $F=\bar{f}$, where $f$ is a ternary operation defined as follows:\begin{align*}
 f(0,0,0)=f(0,0,1)= & f(0,1,1)=f(1,0,1)=0, \\
 f(0,1,0)=f(1,0,0)= & f(1,1,0)=f(1,1,1)=1.
\end{align*}
Obviously, $\bar{f}$ is neither anonymous nor monotone, since e.g. $f(0,0,1)\neq f(0,1,0)$ and $f(0,1,0)=1$, whereas $f(0,1,1)=0$. On the other hand, $\bar{f}$ is StrongDem. Indeed, for each component of $\bar{f}$, it holds that: $$f(x,0,1)=f(0,x,1)=f(0,0,x)=0,$$ for all $x\in\{0,1\}$.

Now, consider $G=\bar{g}$ where $g$ is a ternary operation defined as follows:\begin{align*}
 g(0,0,0)=g(0,0,1)=g(0,1,0) & =g(1,0,0)=g(1,1,0)=0, \\
 g(0,1,1)=g(1,0,1)= & g(1,1,1)=1.
\end{align*}
Again, $\bar{g}$ is easily not anonymous, since $g(1,1,0)\neq g(0,1,1)$. On the other hand, $\bar{g}$ is monotone and StrongDem. For the latter, observe that: $$g(x,0,0)=g(0,x,0)=g(0,0,x)=0,$$ for all $x\in\{0,1\}$. The former is very easy to check and is left to the reader.

Finally, let $H=\bar{h}$, where $h$ is a $4$-ary operation defined as follows:
$$h(x,y,z,w)=1\text{ if and only if exactly two or all of }x,y,z,w\text{ are equal to }1.$$
Since the output of $h$ does not depend on the positions of the input bits, $h$ is anonymous. Also, $h$ is $1$-immune, since:$$h(x,0,0,0)=h(0,y,0,0)=h(0,0,z,0)=h(0,0,0,w),$$ for all $x,y,z,w\in\{0,1\}$. On the other hand, $h$ is not monotone, since $h(0,0,1,1)=1$ and $h(0,1,1,1)=0$.\hfill$\diamond$
\end{example}

The only combination of properties from Definition \ref{def:types} we have not seen, is an anonymous and monotone aggregator that is not Strong Dem. We end this subsection by proving that such aggregators do not exist.
\begin{lemma}\label{a+m=s}
Let $f$ be a $k$-ary anonymous and monotone Boolean function. Then, $f$ is also $1$-immune.
\end{lemma}
\begin{proof}
For $k=2$, the only anonymous functions are $\wedge$ and $\vee$, which are also $1$-immune.

Let $k\geq 3$. Since $f$ is anonymous and monotone, it is not difficult to observe that there is some $l\in\{0,\ldots,k\}$, such that the output of $f$ is $0$ if and only if there are at most $l$ $1$'s in the input bits. If $l>0$, then: $$f(x,0,0\ldots,0,0)=f(0,x,0,\ldots,0,0)=\cdots=f(0,0,0,\ldots,0,x)=0,$$ for all $x\in\{0,1\}$. If $l=0$, then: $$f(x,1,1\ldots,1,1)=f(1,x,1,\ldots,1,1)=\cdots=f(1,1,1,\ldots,1,x)=1,$$ for all $x\in\{0,1\}$. In both cases, $f$ is $1$-immune.\end{proof}

\subsection{Characterizations for domains admitting anonymous, monotone and StrongDem aggregators}\label{ssec:threetypeschar}
We begin with the syntactic characterization of domains admitting anonymous aggregators. Nehring and Puppe \cite[Theorem $2$]{nehring2010abstract} showed that a domain admits a monotone locally non-dictatorial aggregator if and only if it admits a monotone anonymous one. Kirousis et al. \cite{kirousis2017aggregation} strengthened this result by dropping the monotonicity requirement and fixing the arity of the anonymous aggregator, as a direct consequence of Theorem \ref{thm:upd-char}.
\begin{corollary}[Kirousis et al. \cite{kirousis2017aggregation}, Corollary $5.11$]
$D$ is a local possibility domain if and only if it admits a ternary anonymous aggregator. 
\end{corollary}
\begin{proof}
Immediate from the fact that any aggregator of the type described in Theorem \ref{thm:upd-char} is anonymous.
\end{proof}
Thus, we obtain the following result.
\begin{corollary}\label{cor:anonymous}
$D$ admits a $k$-ary anonymous aggregator if and only if there exists a local possibility integrity constraint whose set of models equals $D$.
\end{corollary}

To deal with monotone aggregators we will need some preliminary work. The first fact we will use is that the set of aggregators of a domain $D$ is \emph{closed under superposition}. That is, if $F=(f_1,\ldots,f_m)$ is a $k$-ary aggregator for $D$ and $G^1,\ldots,G^k$ are $l$-ary aggregators for $D$, where $G^i=(g_1^i,\ldots,g_n^i)$, $i=1,\ldots,k$, then $H:=F(G^1,\ldots,G^k)$ is an $l$-ary aggregator for $D$, where $H=(h_1,\ldots,h_n)$ and: $$h_j(x_1,\ldots,x_l)=f_j(g_j^1(x_1,\ldots,x_l),\ldots,g_j^k(x_1,\ldots,x_l)),$$ for all $j=1,\ldots,n$ and for all $x_1,\ldots,x_l\in\{0,1\}$. The proof of this is straightforward and can be found in \cite[Lemma $5.6$]{kirousis2017aggregation}.

This will help us employ a result from the field of Universal Algebra. \emph{Clones} are sets of operations that contain all projections and are closed under superposition (see e.g. Szendrei \cite{szendrei1986clones}). Post \cite{post1941two} provided a complete classification of clones of Boolean operations. This result has already been effectively used by Kirousis et al. \cite{kirousis2017aggregation} in order to obtain the characterizations of possibility and local possibility domains in the Boolean and non-Boolean framework. Here, we will use it analogously.

By the fact that the set of aggregators of a domain $D$ is closed under superposition, we obtain the following result.
\begin{lemma}For a Boolean domain $D\subseteq\{0,1\}^n$, let, for all $j\in\n$:
$$\Cl_j:=\{f\mid\text{There exists an aggregator }F=(f_1,\ldots,f_n)\text{ for D s.t. }f_j=f\},$$
be the set of the $j$-th components of every aggregator for $D$. Then, $\Cl_j$ is a \emph{clone}.
\end{lemma}
The main feature of Post's classification result we will use here is the following (see \cite{bohler2003playing} for an easy to follow presentation):
\begin{lemma}\label{lem:post}
Let $\Cl$ be a clone containing only unanimous functions. Then, either at least one of $\wedge,\vee,\maj,\oplus$ is in $\Cl$, or $\Cl$ contains only projections. 
\end{lemma}

Finally, we will need some definitions. We say that a teranry Boolean operator $g$ is \emph{commutative} if and only if $$g(x,x,y)=g(x,y,x)=g(y,x,x),$$ for all $x,y\in\{0,1\}$. It is not difficult to see that a ternary operator $g$ is commutative if and only if $g\in\{\wedge^{(3)},\vee^{(3)},\maj,\oplus\}$ (again, see \cite[Lemma $5.7$]{kirousis2017aggregation}).

Lastly, we say that a $k$-ary Boolean operation $f$ is \emph{linear}, if there exist constants $c_0,\ldots,c_k\in\{0,1\}$ such that: $$f(x_1,\ldots,x_k)=c_0\oplus c_1x_1\oplus\cdots\oplus c_kx_k,$$ where $\oplus$ again denotes binary addition $\mod 2$. We need two facts concerning linear functions.

\begin{lemma}\label{lem:basics}
Let $f:\{0,1\}^k\mapsto\{0,1\}$ be a linear function and let $c_0,c_1,\ldots,c_k\in\{0,1\}$ such that:$$f(x_1,\ldots,x_k)=c_0\oplus c_1x_1\oplus\cdots\oplus c_kx_k.$$ Then, $f$ is unanimous if and only if $c_0=0$ and there is an odd number pairwise distinct indices $i\in\{1,\ldots,k\}$ such that $c_i=1$.
\end{lemma}
\begin{proof}
The inverse direction is straightforward. For the forward direction, set $x_1=\cdots=x_k=0$. Then, $f(0,\ldots,0)=c_0$ and thus $c_0=0$ since $f$ is unanimous. Finally, assume, to obtain a contradiction, that there is an even number of $c_1,\ldots,c_k$ that are equal to $1$. Set $x_1=\cdots=x_k=1$. Then, it holds that $f(1,\ldots,1)=0$ and $f$ is not unanimous. Contradiction. 
\end{proof}
Since we work only with unanimous functions, from now on we will assume that a linear function satisfies the conditions of Lemma \ref{lem:basics}. This implies also that any linear function has odd arity.

Let again $f:\{0,1\}^k\mapsto\{0,1\}$. We say that $f$ is an \emph{essentially unary function}, if there exists a \emph{unary} Boolean function $g$ and an $i\in\{1,\ldots,k\}$, such that: $$f(x_1,\ldots,x_k)=g(x_i),$$ for all $x_1,\ldots,x_k\in\{0,1\}$ (again, see \cite{bohler2003playing,bohler2004playing,jeavons1995algebraic,jeavons1999determine}). Obviously, the only unanimous such functions are the projections. Note also that any $k$-ary linear functions $f$, with exactly one $i\in\{1,\ldots,k\}$ such that $c_i=1$ is an essentially unary function.

\begin{lemma}\label{lemma:linear}
Let $f:\{0,1\}^k\mapsto\{0,1\}$ be a linear function, $k\geq 3$. Then, either $f$ is an essentially unary function, or it is neither monotone nor $1$-immune.\end{lemma}
\begin{proof}
Let $c_1,\ldots,c_k\in\{0,1\}$ such that:
$$f(x_1,\ldots,x_k)=c_1x_1\oplus\cdots c_kx_k$$ and assume it is not an essentially unary function. Then, there exist at least three pairwise distinct indices $i\in\{1,\ldots,k\}$ such that $c_i=1$. If there are exactly three, $f=\oplus$, which is easily neither monotone, nor $1$-immune.

Assume now that there are at least five pairwise distinct $i\in\{1,\ldots,k\}$ such that $c_i=1$. We will need only four of these indices.

Now, let $j_1,j_2,j_3,j_4\in\{1,\ldots,k\}$ such that $c_{j_1}$, $c_{j_2}$, $c_{j_3}$, $c_{j_4}= 1$. Set $x_{j_1}=x_{j_2}=x_{j_3}=1$ and $x_i=0$, for all $i\in\{1,\ldots,k\}\setminus\{j_1,j_2,j_3\}$. Then, $f(x_1,\ldots,x_k)=1$. By letting $x_{j_4}=1$ too, we obtain $f(x_1,\ldots,x_k)=0$, which shows that $f$ is not monotone.

Finally, to obtain a contradiction, suppose $f$ is $1$-immune. Then, there exist $d_2,\ldots,d_k\in\{0,1\}$ such that: \begin{align*}
f(0,d_2,\ldots,d_k)= & (1,d_2,\ldots,d_k) \Leftrightarrow\\
c_2d_2\oplus\cdots\oplus c_kd_k= & c_1\oplus c_2d_2\oplus\cdots c_kd_k\Leftrightarrow\\
c_1= & 0.
\end{align*}
Continuing in the same way, we can prove that $c_j=0$, for $j=1,\ldots,k$, which is a contradiction.
\end{proof}

We are now ready to prove our results concerning monotone and StrongDem aggregators.
\begin{theorem}\label{thm:monotone}
A domain $D\subseteq\{0,1\}^n$ admits a monotone non-dictatorial aggregator of some arity if and only if it admits a binary non-dictatorial one.
\end{theorem}
\begin{proof}
That a domain admitting a binary non-dictatorial aggregator, admits also a non-dictatorial monotone one is obvious, since all binary unanimous functions are monotone.

For the forward direction, since $D$ admits a monotone non-dictatorial aggregator, it is a possibility domain. Now, to obtain a contradiction, suppose $D$ does not admit a binary non-dictatorial aggregator. Kirousis et al. \cite[Lemma $3.4$]{kirousis2017aggregation} showed that in this case, every $k$-ary aggregator, $k\geq 2$ for $D$ is systematic (the notion of \emph{local monomorphicity} they use corresponds to systematicity in the Boolean framework).

Now, since $D$ contains no binary non-dictatorial aggregators, $\wedge,\vee\notin\Cl_j$, for all $j\in\n$. Thus, by Lemma \ref{lem:post} either $\maj$ or $\oplus$ are contained in $\Cl_j$, for all $j\in\n$ (since the aggregators must be systematic), lest each $\Cl_j$ contains only projections.

Assume that $\overline{\maj}$ is an aggregator for $D$. Then, by Kirousis et al.\cite[Theorem $3.7$]{kirousis2017aggregation}, $D$ admits also a binary non-dictatorial aggregator. Contradiction.

Thus, we also have that $\maj\notin\Cl_j$, $j=1,\ldots,n$. It follows that only $\oplus\in\Cl_j$, $j=1,\ldots,n$. By Post \cite{post1941two}, it follows that for all $j\in\n$, $\Cl_j$ contains only linear functions (see also \cite{bohler2003playing}). By Lemma \ref{lemma:linear}, we obtain a contradiction.
\end{proof}
Thus, by Proposition \ref{prop:projaggr} and Theorem \ref{thm:nonproj}, we obtain the following syntactic characterization.
\begin{corollary}\label{cor:monotone}
$D$ admits a $k$-ary non-dictatorial monotone aggregator if and only if there exists a separable or renamable partially Horn integrity constraint whose set of models equals $D$.
\end{corollary}

To end this subsection, we now consider StrongDem aggregators. Kirousis et al. \cite{kirousis2017aggregation} used what they named the ``diamond'' operator $\diamond$, in order to combine ternary aggregators to obtain new ones whose components are commutative functions (see also Bulatov's \cite[Section $4.3$]{bulatov2016conservative} ``Three Operations Lemma''). Unfortunately, this operator will not suffice for our purposes, so we will also use a new operator we call the \emph{star} operator $\star$.

Let $F=(f_1,\ldots,f_n)$ and $G=(g_1,\ldots,g_n)$ be two $n$-tuples of ternary functions. Define $E:=F\diamond G$ and $H:=F\star G$ to be the $n$-tuples of ternary functions $E=(e_1,\ldots,e_n)$ and $H=(h_1,\ldots,h_n)$ where:
\begin{align*}
    e_j(x,y,z)= & f_j(g_j(x,y,z),g_j(y,z,x),g_j(z,x,y)),\\
    h_j(x,y,z)= & f_j(f_j(x,y,z),f_j(x,y,z),g_j(x,y,z)),
\end{align*} for all $x,y,z\in\{0,1\}$. Easily, if $F$ and $G$ are aggregators for a domain $D$, then so are $E$ and $H$, since they are produced by a superposition of $F$ and $G$. It is easy to notice that $e_j$ is commutative if and only if either $f_j$ or $g_j$ are. Also, under some assumptions for $F$ and $G$, $H$ has no components equal to $\oplus$. In Kirousis et al.\cite[Lemma $5.10$]{kirousis2017aggregation}, it has been proven that, if $F\diamond G=E=(e_1,\ldots,e_n)$, then $e_j\in\{\wedge^{(3)},\vee^{(3)},\maj,\oplus\}$ if and only if either $f_j$ or $g_j\in\{\wedge^{(3)},\vee^{(3)},\maj,\oplus\}$. Furthermore, if $g_j$ is commutative, then $e_j=g_j$, $j=1,\ldots,n$. On the other hand, for the $\star$ operation, we have the following result.

\begin{lemma}\label{lem:oplus}
Let $F=(f_1,\ldots,f_n)$ be an $n$-tuple of ternary functions, such that $f_j\in\{\wedge^{(3)},\vee^{(3)},\maj,\oplus\}$, $j=1,\ldots,n$, and let $J=\{j\mid f_j=\oplus\}$. Let also $G=(g_1,\ldots,g_n)$ be an $n$-tuple of ternary functions, such that $g_j\in\{\wedge^{(3)},\vee^{(3)},\maj\}$, for all $j\in J$. Then, for the $n$-tuple of ternary functions $F\star G:=H=(h_1,\ldots,h_n)$, it holds that:$$h_j\in\{\wedge^{(3)},\vee^{(3)},\maj\},$$ for $j=1,\ldots,n$.
\end{lemma}
\begin{proof}
First, let $j\in\n\setminus J$. Then, $f_j\in\{\wedge^{(3)},\vee^{(3)},\maj\}$ and let $x,y,z\in\{0,1\}$ that are not all equal (lest we have nothing to show since all $f_j$, $g_j$ are unanimous). If $f_j=\wedge^{(3)}$, then easily: $$h_j(x,y,z)=\wedge^{(3)}(\wedge^{(3)}(x,y,z),\wedge^{(3)}(x,y,z),g_j(x,y,z))=\wedge^{(3)}(0,0,g_j(x,y,z))=0,$$ which shows that $h_j=\wedge^{(3)}$. Analogously, we show that $f_j=\vee^{(3)}$ implies that $h_j=\vee^{(3)}$. Finally, let $f_j=\maj$ and let $\maj(x,y,z):=z\in\{0,1\}$. Then: $$h_j(x,y,z)=\maj(\maj(x,y,z),\maj(x,y,z),g_j(x,y,z))=\maj(z,z,g_j(x,y,z))=z,$$ which shows that $h_j=\maj$.

Thus, we can now assume that $J\neq\emptyset$. Let $j\in J$. Then, we have that $f_j=\oplus$ and $g_j\in\{\wedge^{(3)},\vee^{(3)},\maj\}$. Thus, we have that: $$h_j(x,y,z)=\oplus(\oplus(x,y,z),\oplus(x,y,z),g_j(x,y,z))=g_j(x,y,z),$$ from which it follows that $h_j\in\{\wedge^{(3)},\vee^{(3)},\maj\}$. The proof is now complete.
\end{proof}

At last, we are ready to prove our final results.

\begin{theorem}\label{thm:monotone,StrongDem}
A Boolean domain $D\subseteq\{0,1\}^n$ admits a $k$-ary StrongDem aggregator if and only if it admits a ternary aggregator $F=(f_1,\ldots,f_n)$ such that $f_j\in\{\wedge^{(3)},\vee^{(3)},\maj\}$, for $j=1,\ldots,n$.
\end{theorem}
\begin{proof}
It is very easy to see that all the functions in $\{\wedge^{(3)},\vee,\maj\}$ are $1$-immune. Thus, we only need to prove the forward direction of the theorem.

To that end, let $F=(f_1,\ldots,f_n)$ be a $k$-ary StrongDem aggregator for $D$. Then, by Theorem \ref{thm:upd-char}, there exists a ternary aggregator $G=(g_1,\ldots,g_n)$ such that $g_j\in\{\wedge^{(3)},\vee^{(3)},\maj,\oplus\}$ for $j=1,\ldots,n$. Let $J=\{j\mid f_j=\oplus\}$. If $J=\emptyset$, then we have nothing to prove. Otherwise, consider the clones $\Cl_j$, for each $j\in J$.

Suppose now that there exists a $j\in J$, such that $\Cl_j$ contains neither $\wedge$, nor $\vee$, nor $\maj$. By Post's classification of clones of Boolean functions (see \cite{post1941two,bohler2003playing}) and since $\Cl_j$ contains $\oplus$ and only unanimous functions, $\Cl_j$ contains only linear unanimous functions. Again, by Lemma \ref{lemma:linear}, $\Cl_j$ does not contain any $1$-immune function. Contradiction.

Thus, for each $j\in J$, it holds that $\Cl_j$ contains either $\wedge$, or $\vee$ or $\maj$. In the first two cases, $\Cl_j$ obviously contains $\wedge^{(3)}$ or $\vee^{(3)}$ too respectively. Then, it holds that for each $j\in J$ there exists an aggregator $H^j=(h_1^j,\ldots,h_n^j)$, such that $h_j^j\in\{\wedge^{(3)},\vee^{(3)},\maj\}$. Let $J:=\{j_1,\ldots,j_t\}$.

We will now perform a series of iterative combinations between $G$ and the various $H^j$'s, using the $\diamond$ and $\star$ operators, in order to obtain the required aggregator.

First, let $G^j=G\diamond H^j$, for all $j\in J$. By Kirousis et al. \cite[Lemma $5.10$]{kirousis2017aggregation}, we have that $$G_i^j\in\{\wedge^{(3)},\vee^{(3)},\maj,\oplus\},$$ for all $i\in\n$ and $j\in J$. Furthermore, $$G_{j_s}^{j_s}\in\{\wedge^{(3)},\vee^{(3)},\maj\},$$ for $s=1,\ldots,t$. Thus for the aggregator:
$$G^*:=(\cdots((G\star G^{j_1})\star G^{j_2})\star\cdots\star G^{j_t}),$$
we have, by Lemma \ref{lem:oplus}:
$$G_j^*\in\{\wedge^{(3)},\vee^{(3)},\maj\},$$ for $j=1,\ldots,n$, which concludes the proof.
\end{proof}
Recall Definition \ref{def:upic}. We say that a local possibility integrity constraint is $\oplus$-\emph{free}, if $V_2=\emptyset$. Thus, we obtain the following syntactic characterization.
\begin{corollary}\label{cor:monotone,StrongDem}
A Boolean domain $D\subseteq\{0,1\}^n$ admits a $k$-ary StrongDem aggregator if and only if there exists an $\oplus$-free local possibility integrity constraint whose set of satisfying assignments equals $D$.
\end{corollary}

\subsection{Systematic Aggregators}\label{ssec:syst}
We end this work with a discussion concerning systematic aggregators. This is a natural requirement for aggregators from a Social Choice point of view, given that the issues that need to be decided are of the same nature. Recall that $F=(f_1,\ldots,f_n)$ is systematic if $f_1=f_2=\ldots=f_n$.
\begin{definition}\label{def:poly}
Let $D\subseteq\{0,1\}^n$ be a Boolean domain and $f:\{0,1\}^k\mapsto\{0,1\}$ a $k$-ary Boolean operation. $f$ is a \emph{polymorphism} for $D$ (or $f$ \emph{preserves} $D$ or $D$ \emph{is closed under} $f$) if, for all $x^1,\ldots,x^k\in D$:$$(f(x_1),\ldots,f(x_n))\in D,$$ where $x^i=(x_1^i,\ldots,x_n^i)$ and $x_j=(x_j^1,\ldots,x_j^k)$, $i=1,\ldots,k$, $j=1,\ldots,n$.
\end{definition}
The notion of polymorphisms can be found in many standard texts concerning Abstract and Universal Algebra (see e.g. Szendrei \cite{szendrei1986clones}). The following is obvious by considering the definitions of an aggregator and a polymorphism.
\begin{lemma}\label{lem:systpol}
Let $D\subseteq\{0,1\}^n$ be a Boolean domain and $F=\bar{f}$ a systematic $n$-tuple of $k$-ary Boolean functions. Then $F$ is an aggregator for $D$ if and only if $f$ is a polymorphism for $D$.
\end{lemma}

Polymorphisms have been extensively studied in the bibliography and they play a central role in Post's results we discussed in Subsection \ref{ssec:threetypeschar}. These results where also connected with Complexity Theory, where they can be used to provide an alternative proof to the Dichotomy Theorem in the complexity of the \emph{satisfiability problem} (see \cite{bohler2003playing,bohler2004playing,jeavons1995algebraic,jeavons1999determine}). Here, we use a corollary of this Theorem, that can be obtained directly by Post's Lattice, without considering complexity theoretic notions. For an direct algebraic approach, see also Szendrei \cite[Proposition $1.12$]{szendrei1986clones} (by noting that the only Boolean \emph{semi-projections} of arity at least $3$ are projections).

\begin{corollary}\label{cor:pol}
Let $D\subseteq\{0,1\}^n$ be a Boolean domain. Then, either $D$ admits only essentially unary functions, or it is closed under $\wedge$, $\vee$, $\maj$ or $\oplus$.
\end{corollary}

This directly implies that domains admitting non-dictatorial systematic aggregators are either Horn, dual-Horn, bijunctive or affine. We thus immediately obtain the following characterization.
\begin{corollary}\label{cor:syst}
A Boolean domain $D\subseteq\{0,1\}^n$ admits a $k$-ary non-dictatorial systematic aggregator if and only if there exists an integrity constraint which is either Horn, dual Horn, bijunctive or affine, whose set of satisfying assignments equals $D$. 
\end{corollary}

\begin{remark}\label{pd}
Why does $\overline{\maj}$ appears here, although it did not in the characterization of possibility domains (Theorem \ref{thm:DH})? In the Boolean case, a domain admitting $\overline{\maj}$, also admits a binary aggregator $F=(f_1,\ldots,f_n)$, such that $f_j\in\{\wedge,\vee\}$, $j=1,\ldots,n$ (see Kirousis et al.\cite[Theorem $3.7$]{kirousis2017aggregation}). The problem is that this aggregator need not be systematic. In fact, the proof of the aforementioned theorem would produce a systematic aggregator only if $(0,\ldots,0)$ or $(1,\ldots,1)\in D$. 
\end{remark}

Now, what if want to characterize domains admitting some of the various non-dictatorial aggregators we discussed, but requiring also that these aggregators satisfy systematicity? By Theorem \ref{thm:gendict}, we know that Corollary \ref{cor:syst} works for domains admitting systematic aggregators that are not generalized dictatorships too. Furthermore, all the aggregators (resp. integrity constraints) of Corollary \ref{cor:pol} (resp. \ref{cor:syst}) are locally non-dictatorial and anonymous aggregators (resp. lpic's), thus we also have characterizations for domains admitting systematic locally non-dictatorial or anonymous aggregators.

For domains admitting monotone or StrongDem systematic aggregators, we will obtain the result by Lemma \ref{lemma:linear} and Post's Lattice. We will again use the terminology of polymorphisms.
\begin{corollary}\label{cor:charsyst1}
A domain $D\subseteq\{0,1\}^n$ admits a $k$-ary systematic non-dictatorial monotone or StrongDem aggregator if and only if it is closed under $\wedge$, $\vee$ or $\maj$. 
\end{corollary}
\begin{proof}
It is known (and straightforward to see) that the set of polymorphisms of a domain is a clone. Let $\Cl$ be the Boolean clone of polymorphisms of $D$. Since it admits a non-dictatorial aggregator, at least one operator from $\wedge,\vee,\maj,\oplus$ is in $\Cl$. By Lemma \ref{lemma:linear}, this cannot be only $\oplus$.
\end{proof}
Thus, finally, we have the following result.
\begin{corollary}\label{cor:charsyst2}
A Boolean domain $D\subseteq\{0,1\}^n$ admits a $k$-ary systematic non-dictatorial monotone or StrongDem aggregator if and only if there exists an integrity constraint which is either Horn, dual Horn or bijunctive, whose set of satisfying assignments equals $D$. 
\end{corollary}

\section*{Concluding remarks} 
It is known that any  domain on $n$ issues  can be represented either by $n$ formulas $\phi_1, \ldots, \phi_n$ (an agenda), in which case the domain is the set of binary $n$-vectors, the $i$-component of which represents the acceptance or rejection of $\phi_i$ in a consistent way (logic-based approach), or, alternatively, by a single formula $\phi$ of $n$ variables (an integrity constraint), in which case the domain is the set of models of $\phi$. In the former case, there are results, albeit of non-algorithmic nature, that give us conditions on the syntactic form of the $\phi_i$’s, so that the domain accepts a non-dictatorial aggregator. In this work, we give necessary and sufficient conditions on the syntactic form of formulas to be integrity constraints of domains that accept various kinds of non-dictatorial aggregators. For domains that admit non-dictatorial aggregators, or aggregators that are not generalized dictatorships, we call such formulas possibility integrity constraints and furthermore, we show that a subclass of such formulas, the separable and renamable partially Horn formulas, describe domains admitting monotone non-dictatorial aggregators. Then, we show that local possibility domains and domains admitting anonymous aggregators coincide and are described by local possibility integrity constraints, while domains admitting StrongDem aggregators are described by a subclass of local possibility integrity constraints we called $\oplus$-free. Finally, we discuss the corresponding results for systematic aggregators, which are in fact polymorphisms of the domain. Our  results are algorithmic, in the sense that (i) recognizing integrity constraints of the above types can be implemented in time linear in the length of the input formula and (ii) given a domain admitting some of the above non-dictatorial aggregators, a corresponding integrity  constraint, whose number of clauses is polynomial in the size of the domain, can be constructed in time polynomial in the size of the domain. Our proofs draw  from results  in judgment aggregation theory as well from results  about propositional formulas and logical relations.

\section*{Aknowledgements} We are grateful to Bruno Zanuttini  for his comments that improved the presentation and simplified several proofs. Lefteris Kirousis is grateful to Phokion Kolaitis for  initiating him to the  area of Computational Social Choice Theory. We thank Eirini Georgoulaki for her valuable help in the final stages of writing this paper.


\end{document}